\documentclass[10pt,fleqn]{article}
\usepackage{psfrag}
\pdfoutput=1
\usepackage{amsmath}
\usepackage[dvips]{epsfig}
\usepackage{epsfig}
\usepackage{setspace}
\usepackage{amsmath,amssymb}
\usepackage{graphics}
\usepackage[pdftex,colorlinks,bookmarks,pdfpagemode=UseOutlines,linkcolor=black,
pagecolor=black,urlcolor=black,citecolor=black,letterpaper,pageanchor=false]{hyperref}
\usepackage[margin=1in]{geometry}
\usepackage[T1]{fontenc}
\usepackage[authoryear,round]{natbib}
\usepackage{titlesec}
\usepackage{amsthm}
\usepackage{longtable}
\usepackage{xfrac}

\newtheorem{theo}{Theorem}[section]
\newtheorem{lemma}{Lemma}[section]
\newtheorem{df}{Definition}[section]
\newtheorem{cor}{Corollary}[section]
\newtheorem{assump}{Assumption}[section]
\newtheorem{assert}{Assertion}[section]
\newtheorem{remark}{Remark}[section]
\newcommand{\bl}{\begin{lemma}}
\newcommand{\el}{\end{lemma}}
\newcommand{\be}{\begin{equation}}
\newcommand{\ee}{\end{equation}}
\newcommand{\beqn}{\begin{eqnarray}}
\newcommand{\eeqn}{\end{eqnarray}}
\newcommand{\bt}{\begin{theo}}
\newcommand{\et}{\end{theo}}
\newcommand{\bd}{\begin{df}}
\newcommand{\ed}{\end{df}}
\newcommand{\ba}{\begin{assump}}
\newcommand{\ea}{\end{assump}}
\newcommand{\bass}{\begin{assert}}
\newcommand{\eass}{\end{assert}}
\newcommand{\brem}{\begin{remark}}
\newcommand{\erem}{\end{remark}}
\newcommand{\bc}{\begin{cor}}
\newcommand{\ec}{\end{cor}}

\newcommand{\BB}{{\cal B}}

\newcommand{\LL}{{\cal L}}

\newcommand{\NN}{{\cal N}}

\newcommand{\pt}{\tilde{p}}
\newcommand{\Rt}{\tilde{R}}
\newcommand{\St}{\tilde{S}}
\newcommand{\So}{\overline{S}}

\newcommand{\Ut}{\tilde{U}}
\newcommand{\Vt}{\tilde{V}}
\newcommand{\OBB}{\overline{\BB}}
\newcommand{\nv}{\vec{n}}
\newcommand{\fv}{\vec{f}}

\numberwithin{equation}{section}

\long\def\comment#1{}

\title{A new model for multi-commodity macroscopic
modeling of complex traffic networks
\author{Matthew Wright, Gabriel Gomes, Roberto Horowitz, Alex A. Kurzhanskiy}
}

\begin{document}
\maketitle

\begin{abstract}
We propose a macroscopic modeling framework for
a network of roads and multi-commodity traffic.
The proposed framework is based on the Lighthill-Whitham-Richards kinematic
wave theory; more precisely, on its discretization,
the Cell Transmission Model (CTM),
adapted for networks and multi-commodity traffic.
The resulting model is called the Link-Node CTM (LNCTM).

In the LNCTM, we use the fundamental diagram of
an ``inverse lambda'' shape that allows modeling of the
capacity drop and the hysteresis behavior of the traffic state
in a link that goes from free flow to congestion and back.

A model of the node with multiple input and multiple output links
accepting multi-commodity traffic is a cornerstone of the LNCTM.
We present the multi-input-multi-output (MIMO) node model 
for multi-commodity traffic that supersedes previously developed node
models.
The analysis and comparison with previous node models are provided.

Sometimes, certain traffic commodities may choose between 
multiple output links in a node based on the current traffic state
of the node's input and output links. 
For such situations, we propose
a local traffic assignment algorithm that computes how incoming
traffic of a certain commodity should be distributed between output
links, if this information is not known a priori.
\end{abstract}

{\bf Keywords}: macroscopic first order traffic model, first order node model,
multi-commodity traffic, dynamic traffic assignment, congestion games,
managed lanes

\section{Introduction}\label{sec_intro}
\defcitealias{csmps}{Caltrans, 2015}

Traffic simulation models are important tools for traffic engineers and practitioners. As in other disciplines such as climatology, population dynamics, etc., traffic models have helped to deepen our understanding of traffic behavior. They are widely used in transportation planning projects in which capital investments must be justified with simulation-based studies~\citepalias{csmps}.  Recently, with the increased interest in Integrated Corridor Management (ICM) and Decision Support Systems (DSS), traffic models have found a new role in real-time operations. The requirements for models used in the real-time context, both in terms of execution speed and modeling accuracy, are significantly higher than in the planning context. A model used within an on-line traffic prediction framework, for example, must adapt to incidents, weather changes, special events, as well as to the normal day-to-day variations in demand. It must also be capable of capturing all aspects of the transportation system that are relevant to the operator's decision making process. These may include arterials as well as freeways, HOV and HOT lanes, ramp meters, variable speed limits, etc. They may also include multiple modes of travel (buses, trains, cars) and their respective emission signatures. Traffic models are also a central component of state estimation~\citep{pf}, in which an entire ensemble of simulations covering a range of possible states must be executed in real time. 

The ICM and traffic DSS systems that have been deployed in the U.S. are all driven by microsimulation models. These models simulate the interactions between individual vehicles, and are therefore much more computationally expensive than macroscopic models such as the Cell Transmission Model~\citep{daganzo94, daganzo95a}. One of the reasons for this discrepancy is that macroscopic models have lagged behind their commercial microscopic counterparts in their ability to simulate complex configurations such as HOV lanes, arterial/freeway interactions, and multi-commodity flows. 

This paper attempts to narrow the gap by introducing some new features for macroscopic models. These features are gathered in a model we call the Link-Node Cell Transmission Model, which, as its name implies, is a generalization of Daganzo's original CTM. Next we describe each of the new features.

First we introduce a link model that incorporates the notions of capacity drop and hysteresis that have been observed in empirical studies~\citep{capdrop}. The term ``capacity drop'' refers to the observation that the maximum flow that can be measured directly downstream of a bottleneck is lower if the bottleneck is active than if it is not~\citep{capdrop2}. Hysteresis in traffic dynamics refers to the fact that recovery from congestion usually follows a different path in the flow/density plane than the descent into congestion~\citep{hyst}.  Both of these together can be modeled with a special form of the fundamental diagram commonly referred to as the ``inverted lambda'' form~\citep{koshi1983}\footnote{Other names for this shape used by various authors include ``reverse lambda'' and ``backwards lambda.''}.  Of note is that the ``inverted lambda'' destroys the static functional relationship between flow and density, and therefore requires an extra state, which we will call the ``congestion metastate''. 

The second contribution of the paper relates to the treatment of network nodes. Node models determine how congestion will travel and distribute from one link to upstream links. For example, the node model determines whether congestion on a freeway will remain on the mainline or spill through the onramps onto the streets. We present a new node model which unifies the existing models of \citet{tampere11} and \citet{bliemer07}. This model uses input link priority parameters and can be applied to multi-commodity traffic. Input link priorities define how output supply should be allocated for incoming flows. In~\citet{tampere11}, input link priorities were assigned implicitly. Our node model admits arbitrary input link priorities. 

Third, we propose a relaxation of the first-in-first-out principle (FIFO) that is common to most if not all macroscopic models. The FIFO principle states that all vehicles in a link must travel at the same speed, irrespective of their type or destination. This is not a good approximation in the case of multi-lane traffic, in which each lane is allowed to move at its own speed. Applied to a freeway, it means that congestion backing into the mainline from an offramp will immediately block all lanes. In reality vehicles headed for a congested offramp will tend to queue in the slow lane, while leaving the remaining lanes to continuing traffic. The solution to this problem adopted by authors such as~\citet{shiomi2015} is to assign a state variable to every lane in every link. We propose here a simpler, more computationally efficient scheme in which split ratios are adjusted depending on the downstream availability of space. We allow FIFO to be relaxed across destinations, while preserving FIFO across vehicle types. 

The fourth contribution of this article relates to traffic behavior at complex junctions such as access points to managed lanes (e.g., HOV and HOT lanes). Drivers are allowed to choose whether or not to use a managed lane, and they make that choice based on the current observed state of congestion. For example, the incentive to use an HOT lane is greater when the travel time difference between the HOT and the general purpose lanes is large. Thus, a driver's decision to pay for HOT access will depend on the expected travel times and the price. Split ratios in this situation are not known beforehand -- they are endogenous to the model. We develop a split ratio assignment algorithm\footnote{In this paper, the terms ``local [dynamic] traffic assignment'' and ``split ratio assignment'' are used interchangeably.}  for this problem in which the split ratios across an HOV or HOT gate are computed by the simulator. This same algorithm may be adapted to other situations in which turning decisions depend on local conditions, as is the case when arterial traffic routes itself around an obstruction. 

The paper is organized as follows.  Section~\ref{sec_lnctm}, describes the Link-Node Cell Transmission Model (LNCTM). The main feature of this Section is the ``inverse lambda''-shaped fundamental diagram. Section~\ref{sec_node_model} presents our new traffic node model, starting with the analysis of the simpler multiple-input-single-output (MISO) case, then moving to the general multiple-input-multiple-output (MIMO) case. Then, we introduce the relaxed FIFO rule. The proposed node model is used for computing cross-flows into and out of HOV lanes along the freeway. Section~\ref{sec_sr_assignment} describes the split ratio assignment algorithm. Section~\ref{sec_conclusion} concludes the paper by summarizing the results. For the convenience of the reader, the notation used in this paper is summarized in the Appendix.

\section{Link-Node Cell Transmission Model}\label{sec_lnctm}
The LNCTM models traffic flow in a road network consisting of
links $\LL$ and nodes $\NN$,
where links represent stretches of roads, and nodes represent junctions
that connect links.
A node always has at least one input and at least one output link.
A link is called \emph{ordinary} if it has both begin and end nodes.
A link with no begin node is called \emph{origin}, and a link with no end
node is called \emph{destination}.
Origins are links through which vehicles enter the system,
and destinations are links that let vehicles out.

\begin{figure}[htb]
\centering
\includegraphics[width=3in]{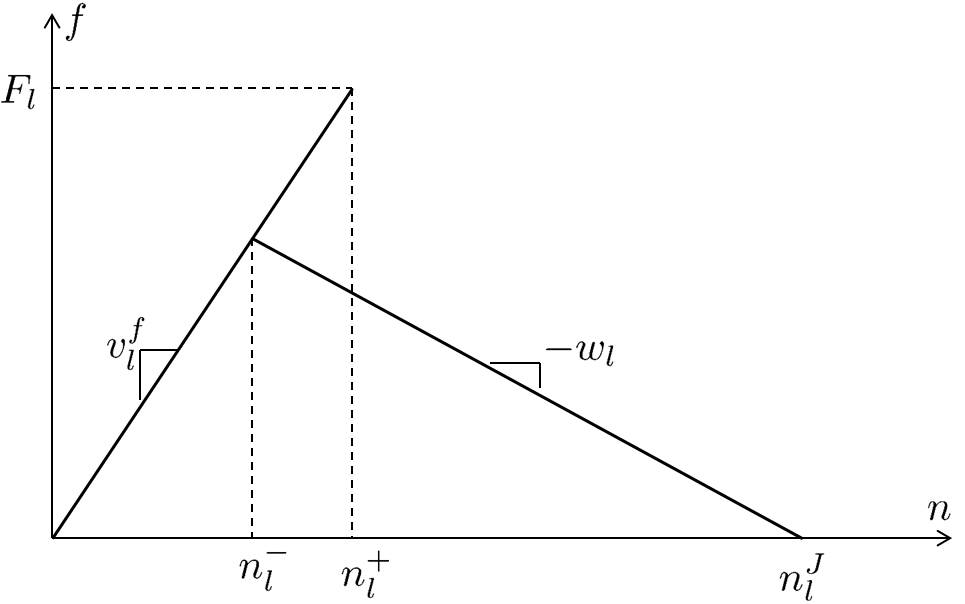}
\caption{Fundamental diagram.}
\label{fig-fd}
\end{figure}

Each link $l\in\LL$ is characterized by its length and
the time dependent fundamental diagram,
a flow-density relationship presented in Figure \ref{fig-fd}.
A fundamental diagram may be \emph{time-varying}
and is defined by 4 values: capacity $F_l(t)$,
free flow speed $v_l^{f}(t)$, congestion wave speed $w_l(t)$ and the jam
density $n^J_l(t)$.
In this paper we assume that densities, flows and speeds are
\emph{normalized} by link lengths and discretization time step;
\footnote{Given original (not normalized) capacity $\tilde{F}_l$
specified in vehicles per hour (vph), free flow speed $\tilde{v}_l^{f}$ and
congestion wave speed $\tilde{w}_l$ specified in miles per hour (mph),
and jam density $\tilde{n}^J_l$ specified in vehicles per mile (vpm),
as well as link length $\Delta x_l$ and discretization time step $\Delta t$,
normalized values are $F_l=\tilde{F}_l\Delta t$ specified in vehicles
per time period $\Delta t$,
$v_l^{f}=\tilde{v}_l^{f}\frac{\Delta t}{\Delta x_l}$ and
$w_l=\tilde{w}_l\frac{\Delta t}{\Delta x_l}$, both unitless, and
$n_l^J = \tilde{n}_l^J \Delta x_l$ specified in vehicles.}
and that free flow speed $v_l^{f}(t)$ and congestion wave speed $w_l(t)$ satisfy
the Courant-Friedrichs-Lewy (CFL) condition \citep{cfl28}:
$0\leq v_l^{f}(t), w_l(t)\leq 1$.
\footnote{The CFL condition is the necessary condition for convergence
while solving hyperbolic PDEs numerically.}
The values $n_l^{-}(t)=\frac{w_l(t)n_J(t)}{v_l^{f}(t) + w_l(t)}$
and $n_l^{+}(t)=\frac{F_l(t)}{v_l^{f}(t)}$ are called \emph{low} and \emph{high
critical density} respectively.
Unless $n_l^{-}(t)=n_l^{+}(t)$, when it assumes triangular shape, 
the fundamental diagram is not a function of density:
$n_l(t)\in\left(n_l^{-}(t), n_l^{+}(t)\right]$ admits two possible flow values.

Each node $\nu\in\NN$ with $M$ input and $N$ output links is characterized by
time dependent mutual restriction coefficients $\{\eta_{jj'}(t)\}^i$,
\footnote{Mutual restriction coefficients determine how flow restrictions
in output links influence each other, when the first-in-first-out (FIFO)
condition us relaxed. This concept is explained in
Section~\ref{subsec_simo}.}
input link priorities $\{p_i(t)\}$ and
partially defined split ratios $\{\beta_{ij}^c(t)\}$,
\footnote{Split ratios may also be fully defined or fully undefined.}
where $C$ is the number of vehicle types;
$i=1,\dots,M$, $j,j'=1,\dots,N$ and $c=1,\dots,C$.

The state of the system at time $t$ is described by the number of vehicles
per commodity in each link:
$\nv_l(t)=\left[n_l^1(t), \dots, n_l^C(t)\right]^T$, where 
$n_l^c(t)$ represents the number of vehicles
of type $c$ in link $l$ at time $t$.
The state update equation for link $l\in\LL$ is:
\be
\nv_l(t+1) = \nv_l(t) + \left(\fv_l^{in}(t) - \fv_l^{out}(t)\right),
\label{eq_link_state_update}
\ee
where $\fv_l^{in}(t) = \left[f_l^{1, in}, \dots, f_l^{C, in}\right]^T$
is the vector of commodity flows coming into link $l$ during this time step,
and $\fv_l^{out}(t) = \left[f_l^{1, out}, \dots, f_l^{C, out}\right]^T$
is the vector of commodity flows leaving link $l$ during this time step.

For ordinary and destination links, $\fv_l^{in}(t)$ is obtained
from the begin node: given a begin node with $M$ input links,
\be
f_l^{c, in}(t) = \sum_{i=1}^M f_{il}^c(t), \;\;\; c=1,\dots,C.
\label{eq_link_f_in}
\ee
For origin links, 
\be
f_l^{c,in}(t) = d_l^c(t),
\label{eq_link_f_in_origin}
\ee
where $d_l^c(t)$ denotes commodity demand at time $t$,
which is an exogenous input to the model, specified in vehicles per
discretization step $\Delta t$.

For ordinary and origin links, $\fv_l^{out}(t)$ is obtained
from the end node: given an end node with $N$ output links,
\be
f_l^{c, out}(t) = \sum_{j=1}^N f_{lj}^c(t), \;\;\; c=1,\dots,C.
\label{eq_link_f_out}
\ee
For destination links, 
\be
f_l^{c, out}(t) = v_l^{f}(t) n_l^c(t) \min\left\{1,
\frac{F_l(t)}{\sum_{c'=1}^C v_l^{f}(t)n_l^{c'}(t)}
\right\}, \;\;\; c=1,\dots, C.
\label{eq_link_f_out_destination}
\ee

The values $f_{il}^c(t)$ and $f_{lj}^c(t)$ are computed by the 
node model that is addressed in detail
in Section~\ref{sec_node_model}.\footnote{More precisely, by the MIMO
algorithm with relaxed FIFO condition described in Section
\ref{subsec_mimo_relaxed}.}

For each link $l\in\LL$ we will also define a congestion metastate:
\be
\theta_l(t) = \left\{\begin{array}{ll}
0 & n_l(t)\leq n^-_l(t),\\
1 & n_l(t)> n^+_l(t),\\
\theta_l(t-1) & n_l^-(t)<n_l(t)\leq n_l^+(t),\end{array}\right.
\label{eq_metastate}
\ee
where $n_l(t)=\sum_{c=1}^C n_l^c(t)$.
This metastate helps determining which constraint of the fundamental
diagram is activated when we compute the receive function for a link.

Now we can formally describe the LNCTM that runs for $T$ time steps.
\begin{enumerate}
\item Initialize:
\begin{eqnarray*}
n_l^c(0) & := & n_{l,0}^c;\\
\theta_l(0) & := & \theta_{l,0};\\
t & := & 0
\end{eqnarray*}
for all $l\in\LL$ and $c=1,\dots,C$, where $n_{l,0}^c$ and $\theta_{l,0}$
are initial conditions.

\item Apply all the control functions that modify system parameters
(fundamental diagrams, input priorities) and/or system state.
A control function may represent ramp metering, variable speed limit,
managed lane policy, etc.
Control functions may be \emph{open-loop} (if they depend only on time)
and \emph{closed-loop} (if they depend on time and system state).
This step is optional.

\item For each link $l\in\LL$ and commodity $c=1,\dots,C$
define the \emph{send} function:
\be
S_l^c(t) = \left\{\begin{array}{ll}
v_l^{f}(t)n_l^c(t)\min\left\{1,
\frac{F_l(t)}{v_l^{f}(t)\sum_{c=1}^C n_l^c(t)}\right\}, &
l \mbox{ is an ordinary link or a destination}, \\
d_l^c(t)\min\left\{1,
\frac{F_l(t)}{\sum_{c=1}^C d_l^c(t)}\right\}, &
l \mbox{ is an origin}.\end{array}\right.
\label{eq_lnctm_send_function}
\ee

\item For each link $l\in\LL$ define the \emph{receive} function:
\be
R_l(t) = \left\{\begin{array}{ll}
\left(1-\theta_l(t)\right)F_l(t) + 
\theta_l(t)w_l(t)\left(n_l^J(t) - \sum_{c=1}^C n_l^c(t)\right), & 
l \mbox{ is an ordinary link or a destination}, \\
\infty, &
l \mbox{ is an origin}.\end{array}\right.
\label{eq_lnctm_receive_function}
\ee

\item For each node $\nu\in\NN$ that has undefined split ratios,
given its input link priorities $\{p_i(t)\}$,
send functions $S_i^c(t)$ and receive functions $R_j(t)$,
compute the undefined split ratios $\{\beta_{ij}^c(t)\}$ according to
the algorithm from Section \ref{sec_sr_assignment}.

\item For each node $\nu\in\NN$,
given its mutual restriction coefficients
$\{\eta_{jj'}^i(t)\}$, input link priorities $\{p_i(t)\}$
and split ratios $\{\beta_{ij}^c(t)\}$,
send functions $S_i^c(t)$ and receive functions $R_j(t)$,
compute input-output
flows $f_{ij}^c$ according to the algorithm from Section \ref{subsec_mimo_relaxed}.

\item For each link $l\in\LL$, compute $\fv_l^{in}(t)$ 
using expressions (\ref{eq_link_f_in})-(\ref{eq_link_f_in_origin})
and $\fv_l^{out}(t)$ using expressions
(\ref{eq_link_f_out})-(\ref{eq_link_f_out_destination}).

\item For each link $l\in\LL$, update the state $\nv_l(t+1)$ according
to the conservation equation (\ref{eq_link_state_update}),
and the metastate $\theta_l(t+1)$ according to its definition
(\ref{eq_metastate}).

\item If $t=T$, then stop, otherwise set $t := t+1$ and return to step 2.
\end{enumerate}

Traffic speed for link $l$ is computed as a ratio of total flow
leaving this link to the total number of vehicles in this link:
\be
v_l(t) = \left\{\begin{array}{ll}
\frac{\sum_{c=1}^C f_l^{c, out}}{\sum_{c=1}^C n_l^c}, &
\mbox{ if } \sum_{c=1}^C n_l^c > 0, \\
v_l^{f}(t), & \mbox{ otherwise}.
\end{array}\right.
\label{eq_speed_formula}
\ee
Defined this way, $v_l(t)\in[0, v_l^f(t)]$.

\section{Node Model}\label{sec_node_model}
As mentioned in Section~\ref{sec_lnctm}, at each timestep of the LNCTM,
traffic flows between links are calculated at the node level,
such that flows through the node are functions of the state of all links
joined at the node (step 5 of the LNCTM algorithm, see
Section~\ref{sec_lnctm}).
For nodes with simple link configurations, such as
the single-chain case with $M=1,N=1$, these calculations are 
straightforward, but the present problem requires a node model formulation 
for computing input-output for general values of $M, N,$ and $C$.

In~\citet{tampere11}, the authors proposed a list of eight requirements
for any general node model.
With our amendments, this list looks as follows.
\begin{enumerate}
\item Applicability to general numbers of input links $M$ and
output links $N$.
We extend this requirement to include 
the general number of traffic commodities $C$.

\item Maximization of the total flow coming out of the node.
According to~\citet{tampere11}, it means that
``each flow should be actively restricted by one of the constraints,
otherwise it would increase until it hits some constraint''.
When a node model is formulated as a constrained optimization problem,
its solution will automatically satisfy this requirement.
However, what this requirement really means is that constraints should be
stated \emph{correctly} and not be overly simplified and, thus,
overly restrictive for the sake of convenient problem formulation.

\item Non-negativity of all input-output flows.

\item Flow conservation: total flow entering the node must be equal to the
total flow exiting.
This requirement is automatically satisfied, since we deal directly with
input-output flows inside the node.

\item Satisfaction of demand and supply constraints.

\item Satisfaction of the first-in-first-out (FIFO) constraint: if a single
destination $j$ for a given $i$ is not able to accept all demand
from $i$ to $j$, all flows from $i$ are also constrained.

We believe that in some situations, the FIFO constraint may be too restrictive.
It should not be completely eliminated, however, but must be relaxed through
a parametrization.
This FIFO relaxation is discribed in Sections~\ref{subsec_simo}
and~\ref{subsec_mimo_relaxed}.
Parameters that we refer to as \emph{mutual restriction coefficients}
allow configuring each pair of node's output links to anything between
no FIFO at all and fully enforced FIFO.

\item Satisfaction of the invariance principle.
If the flow from some input link $i$ is restricted by the available output
supply, this input link enters a congested regime.
This creates a queue in this input link and causes its demand $S_i$ to jump
to capacity $F_i$ in an infinitesimal time, and therefore,
a node model should yield solutions that are invariant to replacing $S_i$ with
$C_i$ when flow from input link $i$ is supply constrained~\citep{lebacque05}.

Although the invariance principle is especially important in developing
numerical methods for continuous time solutions, whereas our traffic
model is discrete in space and time, and its time step $\Delta t$ is
not infinitesimal, we retain it in the list of node model requirements
and will revisit it in Section~\ref{subsec_bliemer_tampere}.

\item Satisfaction of the supply constraint interaction rule (SCIR),
which is supposed to represent the aggregate driver behavior at
congested nodes.
Following~\citet{gentile07}, in~\citet{tampere11} it was proposed to allocate
supply for incoming flows proportionally to input link capacities.

In this paper we propose the concept of input link priorities
that can be represented by arbitrary nonnegative values (or functions)
as SCIR.
Input link priorities will define the allocation of the output supply
for incoming flows.
It must be noted that
the wrong choice of input link priorities may lead to violation of the
invariance principle.
\end{enumerate}

We will add the 9th requirement concerning multi-commodity
nature of modeled traffic to this list.
\begin{enumerate}
\item[9.] Supply restriction on a flow from any given input link is
imposed on commodity components of this flow proportionally to their 
presence in this link.
\end{enumerate}
As opposed to the first eight, we do not view requirement 9
as a general principle.
Rather, we list it here as a rule that we follow throughout this paper.

Further we will provide the mathematical formulation
of the listed requirements and build the node model that satisfies them.
This Section is organized as follows.
We start by describing the merge problem --- the multi-input-multi-output
(MISO) node in Section~\ref{subsec_miso} to explain the concept of input 
link priorities.
From there we move to the general multi-input-multi-output (MIMO) node
in Section~\ref{subsec_mimo}.
Then, we discuss the relationship of the proposed node model with
the node model of~\citet{tampere11} and also compare it with the
node model presented in~\citet{bliemer07} 
in Section~\ref{subsec_bliemer_tampere}.
In Section~\ref{subsec_simo} we propose a way of relaxing the FIFO
condition and explain it in the case of single-input-multi-output (SIMO) node.
Finally, in Section~\ref{subsec_mimo_relaxed} we proceed to the
general MIMO node with relaxed FIFO condition.


\subsection{Multiple-Input-Single-Output (MISO) Node}\label{subsec_miso}
First, we consider a node with $M$ input links and 1 output link.
The number of vehicles of type $c$ that input link $i$ wants to send is $S_i^c$.
The flow entering from input link $i$ has priority $p_i\geq 0$.
Here $i=1,\dots,M$ and $c=1,\dots,C$.
The output link can receive $R$ vehicles.

Casting the computation of input-output flows $f_{i1}^c$ as a mathematical
programming problem, we arrive at:
\be
\max\left(\sum_{i=1}^M\sum_{c=1}^C f_{i1}^c\right), \label{miso_objective}
\ee
subject to:
\begin{eqnarray}
& & f_{i1}^c \geq 0, \;\; i=1,\dots,M, \; c=1,\dots,C \;\; \mbox{ --- non-negativity constraint};
\label{miso_nonnegativity_constraint} \\
& & f_{i1}^c \leq S_i^c, \;\; i=1,\dots,M, \; c=1,\dots,C \;\; \mbox{ --- demand constraint};
\label{miso_demand_constrint} \\
& & \sum_{i=1}^M\sum_{c=1}^C f_{i1}^c \leq R \;\; \mbox{ --- supply constraint};
\label{miso_supply_constraint} \\
& & \frac{f_{i1}^c}{\sum_{c'=1}^C f_{i1}^{c'}} =
\frac{S_{i}^c}{\sum_{c'=1}^C S_{i}^{c'}}, \;\; i = 1,\dots,M,\; c=1,\dots,C,
\;\; \mbox{ --- proportionality constraint} \nonumber\\
& & \mbox{for commodity flows};
\label{miso_proportionality_constraint}\\
& & \left.\begin{array}{cl}
\mbox{(a)} &
p_{i''}\sum_{c=1}^C f_{i'1}^c = p_{i'}\sum_{c=1}^C f_{i''1}^c \;\;
\forall i',i'',c, \mbox{ such that }
f_{i'1}^c<S_{i'}^c, \;
f_{i''1}^c<S_{i''}^c, \\
\mbox{(b)} & 
\mbox{If } \sum_{c=1}^C f_{i1}^c < \sum_{c=1}^C S_i^c, \mbox{ then }
\sum_{c=1}^C f_{i1}^c   \geq
\frac{p_i}{\sum_{i'=1}^M p_{i'}} R. 
\end{array}\right\} \;\; \mbox{ ---} \nonumber\\
& & \mbox{priority constraint}.
\label{miso_priority_constraint}
\end{eqnarray}
Generally,
constraint~\eqref{miso_proportionality_constraint} is optional.
In this paper, however, we assume
that all constraints on input-output flows are imposed
on commodity flows proportionally to commodity contributions to
total demand.
This constraint can be interpreted as a FIFO rule for
multi-commodity traffic within a given link;
not to be confused with the FIFO rule for multiple output links.

Let us discuss priority constraint~\eqref{miso_priority_constraint}
in more detail.
It should be interpreted as follows.
Input links $i=1,\dots,M$, fall into two categories:
(1) those whose flow is restricted by the allocated supply of
the output link; and
(2) those whose demand is satisfied by the supply allocated for them
in the output link.
Priorities define how supply in the output link is allocated for input flows.
Condition~\eqref{miso_priority_constraint}(a)
says that flows from input links of category 1 are allocated
proportionally to their priorities.
Condition~\eqref{miso_priority_constraint}(b)
ensures that input
flows of category 2 do not take up more output supply than was
allocated for them in cases where category 1 is non-empty.
The inequality
in Condition~\eqref{miso_priority_constraint}(b) becomes an
equality when category 2 is empty.

There is a special case, when some input link priorities are
equal to $0$.
If there exists input link $\hat{\imath}$, such that
$p_{\hat{\imath}}=0$, while $f_{\hat{\imath}1}>0$, then,
due to condition~\eqref{miso_priority_constraint}(a),
all input links with non-zero priorities are in category 2.
Thus, if category 1 contains only input links with zero priorities,
one should evaluate condition~\eqref{miso_priority_constraint}(a)
with arbitrary positive, but \emph{equal}, priorities: $p_{i'}=p_{i''}>0$.

If the priorities $p_i$ are proportional to send functions
$\sum_{c=1}^C S_i$, $i=1,\dots,M$,\footnote{Here we assume that $\sum_{c=1}^C S_i^c >0$, $i=1,\dots,M$.} 
then condition~\eqref{miso_priority_constraint} can be written as an equality
constraint:
\be
\frac{\sum_{c=1}^C f_{11}^c}{\sum_{c=1}^C S_1^c} = \dots =
\frac{\sum_{c=1}^C f_{i1}^c}{\sum_{c=1}^C S_i^c} = \dots =
\frac{\sum_{c=1}^C f_{M1}^c}{\sum_{c=1}^C S_M^c},
\label{miso_demand_proportional}
\ee
and the optimization
problem~\eqref{miso_objective}-\eqref{miso_priority_constraint}
turns into a \emph{linear program} (LP).

For arbitrary priorities 
with $M=2$, condition~\eqref{miso_priority_constraint} becomes:
\begin{eqnarray}
& & \sum_{c =1}^C f_{11}^c \leq
\max\left\{\frac{p_1}{p_1+p_2}R,\; R-\sum_{c=1}^C S_2^c\right\};
\label{miso_2in_priority_constraint_1} \\
& & \sum_{c =1}^C f_{21}^c \leq
\max\left\{\frac{p_2}{p_1+p_2}R,\; R-\sum_{c=1}^C S_1^c\right\}.
\label{miso_2in_priority_constraint_2}
\end{eqnarray}
To give a hint how more complicated
constraint~\eqref{miso_priority_constraint} becomes
as $M$ increases,
let us write it out for $M=3$:
\be
\sum_{c=1}^C f_{11}^c \leq
\max\left\{
\frac{p_1}{\sum_{i=1}^3 p_i}R,
\frac{p_1}{p_1+p_2}\left(R-\sum_{c=1}^C S_3^c\right), 
\frac{p_1}{p_1+p_3}\left(R-\sum_{c=1}^C S_2^c\right), 
R-\sum_{i=2,3}\sum_{c=1}^C S_i^c
\right\}; \label{miso_3in_priority_constraint_1} \\
\ee
\be
\sum_{c=1}^C f_{21}^c \leq
\max\left\{
\frac{p_2}{\sum_{i=1}^3 p_i}R,
\frac{p_2}{p_1+p_2}\left(R-\sum_{c=1}^C S_3^c\right), 
\frac{p_2}{p_2+p_3}\left(R-\sum_{c=1}^C S_1^c\right), 
R-\sum_{i=1,3}\sum_{c=1}^C S_i^c
\right\}; \label{miso_3in_priority_constraint_2}
\ee
\be
\sum_{c=1}^C f_{31}^c \leq
\max\left\{
\frac{p_3}{\sum_{i=1}^3 p_i}R,
\frac{p_3}{p_1+p_3}\left(R-\sum_{c=1}^C S_2^c\right), 
\frac{p_3}{p_2+p_3}\left(R-\sum_{c=1}^C S_1^c\right), 
R-\sum_{i=1,2}\sum_{c=1}^C S_i^c
\right\}.\label{miso_3in_priority_constraint_3}
\ee
As we can see, right hand sides of 
inequalities~\eqref{miso_3in_priority_constraint_1}-\eqref{miso_3in_priority_constraint_3}
contain known quantities, and so for arbitrary priorities,
problem~\eqref{miso_objective}-\eqref{miso_priority_constraint} is also an LP.
For general $M$, however, building constraint~\eqref{miso_priority_constraint}
requires a somewhat involved algorithm.
Instead, we present the algorithm for computing input-output flows
$f_{i1}^c$ that solves the maximization
problem~\eqref{miso_objective}-\eqref{miso_priority_constraint}.
\begin{enumerate}
\item Initialize:
\begin{eqnarray*}
\Rt(0) & := & R; \\
U(0) & := & \left\{1,\dots,M\right\}; \\
k & := & 0.
\end{eqnarray*}
$U(k)$ is the set of unprocessed input links:
input links whose input-output flows have not been assigned yet.

\item Check that at least one of the unprocessed input links has
nonzero priority, otherwise, assign equal positive priorities
to all the unprocessed input links:
\[
\pt_i(k) = \left\{\begin{array}{ll}
p_i, &
\mbox{ if there exists } i'\in U(k):\; p_{i'} > 0,\\
\frac{1}{|U(k)|}, &
\mbox{ otherwise},
\end{array}\right.
\]
where $|U(k)|$ denotes the number of elements in set $U(k)$.

\item Define the set of input links
that want to send fewer vehicles than their allocated supply
and whose flows are still undetermined:
\[
\Ut(k) = \left\{i\in U(k):\; \sum_{c=1}^C S_i^c\leq\pt_i(k)
\frac{\Rt(k)}{\sum_{i'\in U(k)}\pt_{i'}(k)} \right\}.
\]
\begin{itemize}
\item If $\Ut(k) \neq\emptyset$, assign:
\begin{eqnarray*}
f_{i1}^c & = & S_i^c, \;\;\; i\in\Ut(k); \\
\Rt(k+1) & = & \Rt(k) - \sum_{i\in\Ut}\sum_{c=1}^C f_{i1}^c; \\
U(k+1) & = & U(k) \setminus \Ut(k).
\end{eqnarray*}
\item Else, assign:
\begin{eqnarray*}
f_{i1}^c & = & S_i^c \frac{\pt_i(k)}{\sum_{i'\in U(k)}\pt_{i'}(k)}
\frac{\Rt(k)}{\sum_{c=1}^C S_i^c},
\;\;\; i\in U(k); \\
U(k+1) & = & \emptyset .
\end{eqnarray*}
\end{itemize}

\item If $U(k+1)=\emptyset$, then stop.

\item Set $k:=k+1$, and return to step 2.
\end{enumerate}

This algorithm finishes after no more than $M$ iterations, and
in the special case of $M=2$ it reduces to
\begin{eqnarray}
f_{11}^c & = & \min\left\{S_1^c, \;\; S_1^c\frac{\max\left\{
\frac{p_1}{p_1+p_2}R, \;\; R-\sum_{c'=1}^CS_2^{c'}\right\}}{
\sum_{c'=1}^CS_1^{c'}}\right\}; \label{eq_merge_21_1}\\
f_{21}^c & = & \min\left\{S_2^c, \;\; S_2^c\frac{\max\left\{
\frac{p_2}{p_1+p_2}R, \;\; R-\sum_{c'=1}^CS_1^{c'}\right\}}{
\sum_{c'=1}^CS_2^{c'}}\right\}. \label{eq_merge_21_2}
\end{eqnarray}

The following theorem can be trivially proved,
and we leave the proof to the reader.
Further, in Section~\ref{subsec_mimo_relaxed}, it will become clear
that this theorem is a special case of a more general statement.
\begin{theo}
The input-output flow computation algorithm constructs the
unique solution of the maximization
problem~\eqref{miso_objective}-\eqref{miso_priority_constraint}.
\label{theo_miso_optimal}
\end{theo}

\textbf{Example.} Suppose we have three inputs and one output, with parameters:

\begin{tabular}{l l}
 $C=1$ & $R = 1000$ \\
 $S_1 = 400$ & $p_1 = \sfrac{1}{3}$ \\
 $S_2 = 500$ & $p_2 = \sfrac{2}{3}$ \\
 $S_3 = 200$ & $p_3 = 0$
\end{tabular}

Our solution algorithm for finding the resulting flows proceeds
as follows:
\begin{align*}
\underline{k=0:} & \\
& 1. \qquad U(0) = \{1,2,3\}; \quad \Rt(0) = 1000; \\
& 2. \qquad \pt_1(0) = \sfrac{1}{3}, \; \pt_2(0) = \sfrac{2}{3}, \; \pt_3(0) = 0 \\
& 3. \qquad \Ut(0) = \{2\}; \quad \boldsymbol{f_2 = 500}; \quad \Rt(1) = 500; \quad U(1) = \{1,3\}  \\
\underline{k=1:} & \\
& 2. \qquad \pt_1(1) = \sfrac{1}{3}, \; \pt_3(1) = 0 \\
& 3. \qquad \Ut(1) = \{1\}; \quad \boldsymbol{f_1 = 400}; \quad \Rt(2) = 100; \quad U(2) = \{3\} \\
\underline{k=2:} & \\
& 2. \qquad \pt_3(2) = 1 \\
& 3. \qquad \Ut(2) = \emptyset; \quad \boldsymbol{f_3 = 100}; \quad U(3) = \emptyset
\end{align*}

\subsection{Multiple-Input-Multiple-Output (MIMO) Node}\label{subsec_mimo}
Now we consider the general case: a node with $M$ input, $N$ output links
and $C$ traffic commodities.
Given are input demands $S_i^c$ and split ratios $\beta_{ij}^c$ specified
per commodity, input priorities $p_i\geq 0$,
and output supply $R_j$.

Define \emph{oriented demand}:
\be
S_{ij}^c = \beta_{ij}^c S_i^c,
\label{eq_oriented_demand}
\ee
and \emph{oriented priorities}:
\be
p_{ij} = p_i \frac{\sum_{c=1}^C S_{ij}^c}{\sum_{c=1}^C S_i^c},
\label{eq_oriented_priorities0}
\ee
where $i=1,\dots,M$, $j=1,\dots,N$ and $c=1,\dots,C$.

As before, we start by casting the allocation of input-output flows
$f_{ij}^c$ as a constrained optimization problem:
\be
\max\left(\sum_{i=1}^M\sum_{j=1}^N\sum_{c=1}^C f_{ij}^c\right),
\label{mimo_objective}
\ee
subject to:
\begin{eqnarray}
& & f_{ij}^c \geq 0, \;\; i=1,\dots,M, \; j=1,\dots,N, \; c=1,\dots,C \;
\mbox{ --- non-negativity constraint};
\label{mimo_nonnegativity_constraint} \\
& & f_{ij}^c \leq S_{ij}^c, \;\; i=1,\dots,M, \; j=1,\dots,N, \; 
c=1,\dots,C \; \mbox{ --- demand constraint};
\label{mimo_demand_constraint} \\
& & \sum_{i=1}^M\sum_{c=1}^C f_{ij}^c \leq R_j, \;\; j=1,\dots,N \;
\mbox{ --- supply constraint};
\label{mimo_supply_constraint} \\
& & \frac{f_{ij}^c}{\sum_{c'=1}^C f_{ij}^{c'}} =
\frac{S_{ij}^c}{\sum_{c'=1}^C S_{ij}^{c'}}, \;\; i = 1,\dots,M,\;
j=1,\dots,N, c=1,\dots,C \;\; \mbox{ --- proportionality}\nonumber \\
& & \mbox{constraint for commodity flows};
\label{mimo_proportionality_constraint}\\
& & \frac{\sum_{c=1}^C f_{ij}^c}{\sum_{j=1}^N\sum_{c=1}^C f_{ij}^c} =
\frac{\sum_{c=1}^C S_{ij}^c}{
\sum_{j=1}^N\sum_{c=1}^C S_{ij}^c},
\;\; i=1,\dots,M \;\;  \mbox{ --- FIFO constraint};
\label{mimo_fifo_constraint} 
\end{eqnarray}
\be
\left.\begin{array}{cl}
\mbox{(a)} &
\mbox{For each input link $i$, such that } \\
& \sum_{j=1}^N \sum_{c=1}^C f_{ij}^c < \sum_{c=1}^C S_i^c, \\
& \mbox{there exists an output link $j^\ast$, such that } \\
& p_{i'j^\ast}\sum_{c=1}^C f_{ij^\ast}^c \geq 
p_{ij^\ast}\sum_{c=1}^C f_{i'j^\ast}^c
\forall i'\neq i, \\
& \mbox{and it is said that output $j^\ast$ restricts input flow $i$};\\
\mbox{(b)} &
\mbox{For each input link $i$ whose flow is restricted by some output $j$,} \\
& \sum_{c=1}^C f_{ij}^c \geq \frac{p_{ij}}{\sum_{i'=1}^M p_{i'j}} R_j.
\end{array}\right\} \mbox{ --- priority constraint}.
\label{mimo_priority_constraint}
\ee
Here, the objective function~\eqref{mimo_objective} and
constraints~\eqref{mimo_nonnegativity_constraint}-\eqref{mimo_proportionality_constraint}
are straight forward extensions of the objective
function~\eqref{mimo_objective} and
constraints~\eqref{miso_nonnegativity_constraint}-\eqref{miso_proportionality_constraint}.
The FIFO constraint~\eqref{mimo_fifo_constraint} ensures that, for
a given input link $i$, if
any one of the output links is unable to accommodate its allocation
of flow, all outflow from $i$ is restricted
proportionally~\citep{daganzo94, daganzo95a}.

Let us focus the attention on the priority
constraint~\eqref{mimo_priority_constraint}.
Condition~\eqref{mimo_priority_constraint}(a) says that if the
flow from a given input link $i$ is reduced by the output supply,
there must be the most restrictive output link, which is denoted $j^\ast$.
If the output $j^\ast$ restricts flows from input links $i'$ and $i''$,
then, according to condition~\eqref{mimo_priority_constraint}(a),
\[
p_{i''j^\ast}\sum_{c=1}^C f_{i'j^\ast}^c = 
p_{i'j^\ast}\sum_{c=1}^C f_{i''j^\ast}^c.
\]
For each output link $j$, $j=1,\dots,N$, 
flows $\sum_{c=1}^C f_{ij}^c$ fall into two categories:
(1) restricted by \emph{this} output link; and
(2) not restricted by \emph{this} output link.
Condition~\eqref{mimo_priority_constraint}(a) states that input-output
flows of category 1 are allocated proportionally to their priorities.
For every output link $j$, $j=1,\dots,N$, 
condition~\eqref{mimo_priority_constraint}(b) ensures that
input-output flows of category 2 
do not take up more output supply than was allocated for them when category 1
is non-empty.

Recalling the MISO case, we can see that for $N=1$,
formulation~\eqref{mimo_objective}-\eqref{mimo_priority_constraint}
translates to 
formulation~\eqref{miso_objective}-\eqref{miso_priority_constraint}.

{\bf Remark.}
It must be noted here that when
priorities $p_i$ are proportional to send functions
$\sum_{c=1}^C S_i$, $i=1,\dots,M$, and $M>1$, the 
condition~\eqref{mimo_priority_constraint}
\emph{cannot} be written simply as an extension
of~\eqref{miso_demand_proportional} for any output link $j$,
$j=1,\dots,N$, as was done by Bliemer in~\citeyearpar{bliemer07}. 
He used the constraint
\be
\frac{\sum_{c=1}^C f_{1j}^c}{\sum_{c=1}^C S_{1j}^c} = \dots =
\frac{\sum_{c=1}^C f_{ij}^c}{\sum_{c=1}^C S_{ij}^c} = \dots =
\frac{\sum_{c=1}^C f_{Mj}^c}{\sum_{c=1}^C S_{Mj}^c}, \;\;
j=1,\dots,N,
\label{mimo_demand_proportional}
\ee
in the place of~\eqref{mimo_priority_constraint},
which reduces the optimization
problem~\eqref{mimo_objective}-\eqref{mimo_priority_constraint} to an LP.
While suitable for $N=1$,
expression~\eqref{mimo_demand_proportional}
is not an adequate representation of the
constraint~\eqref{mimo_priority_constraint} for $N>1$, 
because it leads to the underutilization of the available supply,
as will be shown in Section~\ref{subsec_bliemer_tampere}.

Next, we present the algorithm for constructing
the unique solution of the optimization
problem~\eqref{mimo_objective}-\eqref{mimo_priority_constraint}.
\begin{enumerate}
\item Initialize:
\begin{eqnarray*}
\Rt_j(0) & := & R_j; \\
U_j(0) & := & \left\{i\in\{1,\dots,M\}: \;
\sum_{c=1}^C S_{ij}^c > 0 \right\}; \\
k & := & 0; \\
& & i = 1,\dots,M, \;\;\; j = 1,\dots,N, \;\;\; c = 1,\dots,C.
\end{eqnarray*}
$U_j(k)$ is the set of input links contributing traffic to output link $j$,
whose input-output flows are still unassigned.

\item Define the set of output links that still need processing:
\[ V(k) = \left\{j\in\{1,\dots,N\}: \; U_j(k) \neq \emptyset\right\}. \]
If $V(k) = \emptyset$, stop.

\item Check that at least one of the unprocessed input links has
nonzero priority, otherwise, assign equal positive priorities
to all the unprocessed input links:
\be
\pt_i(k) = \left\{\begin{array}{ll}
p_i, &
\mbox{ if there exists } i' \in \bigcup_{j\in V(k)} U_j(k):\;
p_{i'} > 0, \\
\frac{1}{\left|\bigcup_{j\in V(k)}U_j(k)\right|}, &
\mbox{ otherwise},
\end{array}\right.
\label{eq_nonzero_priorities}
\ee
where $\left|\bigcup_{j\in V(k)} U_j(k)\right|$ denotes the number
of elements in the union $\bigcup_{j\in V(k)} U_j(k)$;
and for each output link $j\in V(k)$ and input link $i\in U_j(k)$
compute oriented priority:
\be
\pt_{ij}(k) = \pt_i(k) \frac{\sum_{c=1}^C S_{ij}^c}{\sum_{c=1}^C S_i^c}.
\label{eq_oriented_priorities}
\ee

\item For each $j\in V(k)$, compute factors:
\be
a_j(k) = \frac{\Rt_j(k)}{\sum_{i\in U_j(k)}\pt_{ij}(k)},
\label{eq_aj_factors}
\ee
and find the smallest of these factors:
\be
a_{j^\ast}(k) = \min_{j\in V(k)} a_j(k).
\label{eq_aj_min}
\ee
The link $j^\ast$ has the most restricted supply of
all output links.

\item Define the set of input links,
whose demand does not exceed the allocated supply:
\[
\Ut(k) = \left\{i\in U_{j^\ast}(k):\;
\sum_{c=1}^C S_i^c \leq \pt_i(k)a_{j^\ast}(k)\right\}.
\]
\begin{itemize}
\item If $\Ut(k) \neq \emptyset$, then for all output links $j\in V(k)$ assign:
\begin{eqnarray*}
f_{ij}^c & = & S_{ij}^c, \;\;\; i\in\Ut(k),\;\; c=1,\dots,C;\\
\Rt_j(k+1) & = & \Rt_j(k) - \sum_{i\in\Ut(k)}\sum_{c=1}^C f_{ij}^c; \\
U_j(k+1) & = & U_j(k) \setminus \Ut(k). 
\end{eqnarray*}
\item Else, for all input links $i\in U_{j^\ast}(k)$
and output links $j\in V(k)$ assign:
\begin{eqnarray}
f_{ij}^c & = &
S_{ij}^c \frac{\pt_{ij}(k)a_{j^\ast}(k)}{\sum_{c'=1}^C S_{ij}^{c'}},
\;\;\; i\in U_j(k) \cap U_{j^\ast}(k); \label{eq_mimo_step6}\\
\Rt_j(k+1) & = & \Rt_j(k) - \sum_{i\in U_{j^\ast}(k)}\sum_{c=1}^C f_{ij}^c;
\nonumber\\
U_j(k+1) & = & U_j(k) \setminus U_{j^\ast}(k), \nonumber
\end{eqnarray}
where, obviously, $U_{j^\ast}(k+1) = \emptyset$.
\end{itemize}

\item Set $k:=k+1$, and return to step 2.
\end{enumerate}

This algorithm takes no more than $M$ iterations to complete.

The following lemma states that in the case of $N=1$,
the MIMO algorithm produces exactly the same result as the MISO
algorithm described in Section \ref{subsec_miso}.
\begin{lemma}
The MISO algorithm is a special case of the MIMO algorithm with $N=1$.
\label{lemma_mimo}
\end{lemma}
\begin{proof}
The proof follows from the fact that for $N=1$,
formulae (\ref{eq_aj_factors})-(\ref{eq_aj_min}) result in
$a_{j^\ast}(k) = a_1(k)=\frac{\Rt_{1}(k)}{\sum_{i\in U_1(k)}\pt_i(k)}$ and $j^\ast = 1$.
\end{proof}

The following theorem is the analog of Theorem~\ref{theo_miso_optimal}.
As in the MISO case, we leave the proof out, since a more general
statement will be proved in Section~\ref{subsec_mimo_relaxed}.
\begin{theo}
The input-output flow computation algorithm constructs the
unique solution of the maximization
problem~\eqref{mimo_objective}-\eqref{mimo_priority_constraint}.
\label{theo_mimo_optimal}
\end{theo}

\subsection{Relationship with Tamp\'{e}re Et Al. and Bliemer Node Models}
\label{subsec_bliemer_tampere}

It can be shown that the node model described here
is a slight generalization of the Tamp\'{e}re et al. unsignalized
node model developed in~\citet{tampere11}, 
In particular, for a specific choice of priorities, $p_i$, 
the equations of section~\ref{subsec_mimo} become exactly the equations of this model. 

\begin{theo}
  By setting the priorities for the MIMO model
  equal to the link capacity, $F$,
  \begin{equation*}
    p_{i} = F_i,
  \end{equation*}
  with $F_i$ the capacity of link $i$, we obtain the Tamp\'{e}re et al. node model.
  \label{theo-tampere}
\end{theo}

\begin{proof}
  The Tamp\'{e}re et al. model does not deal with multiple commodities, so first,
  let $C=1$ and drop the commodity index $c$. Then, our oriented
  priorities~\eqref{eq_oriented_priorities} become:
  \begin{equation*}
    \pt_{ij} = F_i \frac{S_{ij}}{S_i}.
  \end{equation*}

  Again for simplicity, assume that we do not have any input links whose demands can be 
  wholly met. Then, the smallest $a_j(k)$ factor is
  \begin{equation*}
    a_{j^\ast} = \min_{j \in V(k)} \frac{\tilde{R}_j}{ \sum_{i \in U_j(k)} F_{i} \frac{S_{ij}}{S_i}}
  \end{equation*}
  and the resulting flows are
  \begin{equation}
    f_{ij} = S_{ij} \frac{\pt_{ij}(k) a_{j^\ast}(k)}{S_{ij}} = \pt_{ij}(k) a_{j^\ast}(k) = F_i \frac{S_{ij}}{S_i} \frac{ \Rt_{j^\ast}}{ \sum_{i \in U_{j^\ast}} F_i \frac{S_{ij^\ast}}{S_i}}
  \end{equation}
  comparing the above development to equations (19), (24), (26), and (28) in
  \citet{tampere11}, we have recovered exactly the Tamp\'{e}re et al. node model.
\end{proof}

\begin{figure}
\centering
\includegraphics[width=1.5in]{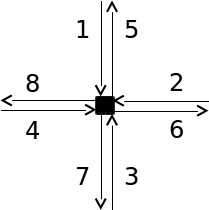}
\caption{A 4-by-4 example intersection.}
\label{fig-mimo_example}
\end{figure}

\textbf{Example.} This specific example is taken from Tamp\'{e}re et
al.~\citeyearpar{tampere11}, to illustrate both the procedure of our MIMO node
algorithm and the equivalence of our node model to the Tamp\'{e}re et al.
model proved above.

Consider the 4-by-4 intersection shown in Figure~\ref{fig-mimo_example}.
Suppose that our relevant parameters are:

\begin{tabular}{l l l l}
 $C=1$ & & & \\
 $S_1=500$ & $S_2=2000$ & $S_3=800$ & $S_4=1700$ \\
 $p_1=F_1=1000$ & $p_2=F_2=2000$ & $p_3=F_3=1000$ & $p_4=F_4=2000$ \\
 $R_5=1000$ & $R_6=2000$ & $R_7=1000$ & $R_8=2000$ \\
 $\beta_{15} = 0$ & $\beta_{16} = 0.1$ & $\beta_{17} = 0.3$ & $\beta_{18} = 0.6$ \\
 $\beta_{25} = 0.05$ & $\beta_{26} = 0$ & $\beta_{27} = 0.15$ & $\beta_{28} = 0.8$ \\
 $\beta_{35} = 0.125$ & $\beta_{36} = 0.125$ & $\beta_{37} = 0$ & $\beta_{38} = 0.75$ \\
 $\beta_{45} = \sfrac{1}{17}$ & $\beta_{46} = \sfrac{8}{17}$ & $\beta_{47} = \sfrac{8}{17}$ & $\beta_{48} = 0$
\end{tabular}

Our solution algorithm proceeds as follows:

First, compute oriented demands:

\begin{tabular}{l l l l}
 $S_{15}=0$ & $S_{16} = 50$ & $S_{17}=150$ & $S_{18} = 300$ \\
 $S_{25}=100$ & $S_{26} = 0$ & $S_{27}=300$ & $S_{28} = 1600$ \\
 $S_{35}=100$ & $S_{36} = 100$ & $S_{37}=0$ & $S_{38} = 600$ \\
 $S_{45}=100$ & $S_{46} = 800$ & $S_{47}=800$ & $S_{48} = 0.$
\end{tabular}

\underline{$k=0:$}
\begin{enumerate}
	\item $U_5(0) = \{2,3,4\}; \quad U_6(0) = \{1,3,4\}; \quad U_7(0) = \{1,2,4\}; \quad U_8(0) = \{1,2,3\} $
	\item $ V(0) = \{5,6,7,8\}$
	\item No adjustments to the $\pt_{i}$ are made. The oriented priorities are: \\
		\begin{tabular}{l l l l}
		$\pt_{15} = 0$ & $\pt_{16} = 100$ & $\pt_{17}=300$ & $\pt_{18}=600$ \\
			$\pt_{25} = 100$ & $\pt_{26} = 0$ & $\pt_{27}=300$ & $\pt_{28}=1600$ \\
			$\pt_{35} = 125$ & $\pt_{36} = 125$ & $\pt_{37}=0$ & $\pt_{38}=750$ \\
			$\pt_{45} = 118$ & $\pt_{46} = 941$ & $\pt_{47}=941$ & $\pt_{48}=0$
		\end{tabular}
	\item $a_5(0) = 2.92; \quad a_6(0) = 1.72; \quad a_7(0) = 0.649; \quad a_8(0) = 0.678$ \\
		$a_{j^\ast}(0) = a_7(0) = 0.649$
	\item $\Ut(0) = \{1\}, \text{ as } S_1 = 500 \leq \pt_1(0) a_7(0) = 1000 \times 0.649$
		\begin{itemize}
			\item $\boldsymbol{f_{16}=S_{16}=50;} \quad \boldsymbol{f_{17}=S_{17}=150;} \quad \boldsymbol{f_{18}=S_{18}=300;}$
			\item $\Rt_6(1) = 2000-50=1950; \quad \Rt_7(1)=1000-150=850; \quad \Rt_8(1)=2000-300=1700$
			\item $U_5(1) = \{2,3,4\}; \quad U_6(1)=\{3,4\}; \quad U_7(1)=\{2,4\}; \quad U_8(1)=\{2,3\}$
		\end{itemize}
\end{enumerate}

\underline{$k=1:$}
\begin{enumerate}
\setcounter{enumi}{1}
	\item $V(1) = \{5,6,7,8\}$
	\item No changes are made to priorities or oriented priorities.
	\item $a_5(1)=2.92; \quad a_6(1) = 1.83; \quad a_7(1) = 0.685; \quad a_8(0) = 0.723$ \\
		$a_{j^\ast}(1) = a_7(1) = 0.685$
	\item $\Ut(1) = \emptyset,$ as $S_2=2000 \nleq \pt_2(1) a_7(1) = 1370,$ and $S_4=1700 \nleq \pt_4(1) a_7(1) = 1370,$
	\begin{itemize}
		\item $\boldsymbol{f_{25} = \pt_{25}(1)a_7(1) = 68.5}$
		\item $\boldsymbol{f_{27} = \pt_{27}(1)a_7(1) = 205.5}$
		\item $\boldsymbol{f_{28} = \pt_{28}(1)a_7(1) = 1096}$
		\item $\boldsymbol{f_{45} = \pt_{45}(1)a_7(1) = 80.6}$
		\item $\boldsymbol{f_{46} = \pt_{46}(1)a_7(1) = 644.5}$
		\item $\boldsymbol{f_{47} = \pt_{47}(1)a_7(1) = 644.5}$
		\item $\Rt_5(2)=1000-68.5-80.6=850.9; \quad \Rt_6(2) = 1950-644=1305.5;$ \\
			$\Rt_7(2)=850-644.5-204.5=0; \quad \Rt_8(2)=1700-1096=604$
		\item $U_5(2) = \{3\}; \quad U_6(2)=\{3\}; \quad U_7(2)=\emptyset; \quad U_8(2)=\{3\}$
	\end{itemize}
\end{enumerate}

\underline{$k=2:$}
\begin{enumerate}
\setcounter{enumi}{1}
	\item $V(2) = \{5,6,8\}$
	\item No changes are made to priorities or oriented priorities.
	\item $a_5(2) = 6.81; \quad a_6(2) = 10.45; \quad a_8(2) = 0.805$ \\
		$a_{j^\ast}(2) = a_8(2) = 0.805$
	\item $\Ut(2) = \{3\},$ as $S_3 = 800 \leq \pt_3(2) a_8(2) = 805$
	\begin{itemize}
		\item $\boldsymbol{f_{35}=S_{35}=100;} \quad \boldsymbol{f_{36}=S_{36}=100;} \quad \boldsymbol{f_{38}=S_{38}=600;}$
		\item $U_j(3) = \emptyset \quad \forall j$
	\end{itemize}
\end{enumerate}

\underline{$k=3:$}
\begin{enumerate}
\setcounter{enumi}{1}
	\item $V(3) = \emptyset$, so the algorithm terminates. All flows have been found.
\end{enumerate}

As discussed in Section~\ref{subsec_mimo}, for $N>1$ the demand-proportional priority
node model of Bliemer~\citeyearpar{bliemer07} cannot be written as a special case of this
node model, and that replacing the priority constraint~\eqref{mimo_priority_constraint}
with a constraint of the form~\eqref{mimo_demand_proportional} leads to underused supply.
Let us consider a simple example:

\textbf{Example.} Suppose we have a 2-by-2 intersection, with parameters

\begin{tabular}{l l l l}
 $C=1$ & & $p_1=0.5$ & $p_2=0.5$ \\
 $S_1=1000$ & $S_2=1000$ & $R_1=600$ & $R_2=1000$ \\
 $\beta_{11} = \sfrac{9}{10}$ & $\beta_{12} = \sfrac{1}{10}$ & & \\
 $\beta_{21} = 0$ & $\beta_{22} = 1$ & & \\
\end{tabular}

It should be clear that our algorithm will solve this problem in two iterations.
In the first iteration, we will have
\begin{align*}
	a_1(0) &= \sfrac{600}{900} = \sfrac{2}{3} = a_{j^\ast}(0) \\
	a_2(0) &= \sfrac{1000}{1100} = \sfrac{10}{11}
\end{align*}
and, as $U_{j^\ast}(0) = U_{1}(0) = \{1\}$, the flows found will be
\begin{align*}
	f_{11} &= 600 \\
	f_{12} &= \sfrac{2}{3} \times 100 \approx 67.
\end{align*}
In the second iteration, link $i=2$ will fill up all remaining space in link $j=2$:
\begin{align*}
	f_{21} &= 0 \\
	f_{22} &= R_2 - f_{12} \approx 933,
\end{align*}
and the total flow passing through the node is
\[
f_{11}+f_{12}+f_{21}+f_{22} = 1600.
\]

On the other hand, in the Bliemer model we first encounter
the supply restriction $R_1$ that, together with the FIFO
constraint~\eqref{mimo_fifo_constraint}, reduces the flow
from input link 1 by the factor $\sfrac{2}{3}$, and we get
\begin{align*}
	f_{11} &= 600 \\
	f_{12} &= \sfrac{2}{3} \times 100 \approx 67,
\end{align*}
just as in our calculation above.
There is no flow from input link 2 to output link 1: $f_{21}=0$.
Next, condition~\eqref{mimo_demand_proportional} implies that
\[
\frac{f_{11}}{S_{11}}=
\frac{f_{12}+f_{22}}{S_{12}+S_{22}}=\frac{600}{900},
\]
which yields:
\[
f_{22} = \frac{2}{3}(S_{12}+S_{22}) - f_{21} \approx 667,
\]
and so, the total flow through the node is
\[
f_{11}+f_{12}+f_{21}+f_{22} = 1334.
\]
As we see, the Bliemer model violates requirement 2 from the list in the
beginning of Section~\ref{sec_node_model}: it does not maximize the total
flow through the node.
Failure to maximize the flow happened in the LP setting of~\citet{bliemer07}
due to the \emph{incorrect constraint}~\eqref{mimo_demand_proportional}
that was imposed in pursuit of an easy formulation of the optimization
problem.

It was shown in~\citet{tampere11} that the Bliemer model also violates the
invariance principle --- requirement 7 from the list.
The violation of the invariance principle happened because supply allocation
was governed by the input demand.
The input link started to congest, and the demand jumped in infinitely small
time, which immediately triggered change in the supply allocation
of the output link, ultimately resulting in discontinuous flow.
In a discrete time system, such as LNCTM, these considerations
may be simply ignored.
However, since our proposed node model is generic and can be used outside
of the LNCTM, we should note that to avoid violating the invariance principle,
input link priorities must not depend on traffic state.

\subsection{Single-Input-Multiple-Output (SIMO) Node with Relaxed FIFO Condition}
\label{subsec_simo}
Sometimes the FIFO rule may be too restrictive.
An example of such situation is a junction, where a 1-lane off-ramp
diverges from a 5-lane freeway, especially, if this off-ramp
is relatively short and terminates with a signal.
Jammed off-ramp blocks the flow on the whole freeway, which is not realistic.
A more realistic behavior would be for this jammed off-ramp
to only partially restrict the mainline flow.
For instance, one could say that this off-ramp restricts only 40\% of the
mainline traffic.
On the other hand, if the mainline happens to be jammed,
we want 100\% of the off-ramp traffic to be restricted.
Thus, in a given diverge junction FIFO restrictions on outgoing
flows are not symmetric.
The model of a diverge node with a partial FIFO restriction is described next.

Consider a node with a single input and $N$ output links.
Given are the input demand per commodity, $S^c$, $c=1,\dots,C$,
the output supply $R_j$, split ratios $\beta_{1j}^c$, $j=1,\dots,N$,
and the mutual restriction matrix
$\{\eta_{jj'}\}\in \mathbb{R}^{N\times N}$, $\eta_{jj'}\in[0, 1]$, $\eta_{jj}=1$.
Element $\eta_{jj'}$ of the mutual restriction matrix specifies the portion
of the flow in the output link $j'$ affected by the restriction of the output link $j$.
In the example above, where the mainline output link is identified as 1
and the off-ramp as 2, this matrix assumes the form
$\left[\begin{array}{cc}1 & 1\\
0.4 & 1\end{array}\right]$.

To formulate the optimization problem for the SIMO case,
we can re-use the objective function~\eqref{mimo_objective}
and constraints~\eqref{mimo_nonnegativity_constraint}-\eqref{mimo_proportionality_constraint}
directly, with $M=1$;
constraint~\eqref{mimo_priority_constraint} can be dropped, since
there is no competition for the output supply between incoming flows;
and the FIFO constraint~\eqref{mimo_fifo_constraint} has to be replaced.

Using~\eqref{eq_oriented_demand}, we obtain oriented demand:
\be
S_{1j}^c = \beta_{1j}^c S^c, \;\;\; j=1,\dots,N.
\label{eq_oriented_demand2}
\ee
The relaxed FIFO constraint can now be written:
\be
\sum_{c=1}^C f_{1j}^c\leq
\left(1-\eta_{j'j}\right)\sum_{c=1}^C S_{1j}^c +
\eta_{j'j}
\left(\frac{R_{j'}}{\sum_{c=1}^C S_{1j'}^c}\right)
\left(\sum_{c=1}^C S_{1j}^c\right), \;
j=1,\dots,N, \; \forall j'\neq j.
\label{simo_rfifo_constraint}
\ee
The right hand side of the last inequality consists of two terms.
The first term represents the portion of the oriented demand
unaffected by a possible supply shortage in the output link $j'$.
If the full FIFO rule is enforced, that is, $\eta_{j'j}=1$, then this
term equals 0.
The second term represents the portion of the oriented demand influenced
by the output $j'$.
The necessary, but not sufficient, condition for activation 
of this constraint is
$\frac{R_{j'}}{\sum_{c=1}^C S_{1j'}^c} < 1$
for some $j'$.
If the FIFO rule is abandoned, that is, $\eta_{j'j}=0$, then
the second term equals 0.

The algorithm for computing input-output flows $f_{1j}^c$ follows.
\begin{enumerate}
\item For each output $j$, compute the reduction factor:
\be
\alpha_{j} = \min\left\{1, \frac{R_j}{\sum_{c=1}^C S_{1j}^c}\right\}
\label{eq_reduction_factor}
\ee
The factor $\alpha_{j}$ defines by how much the demand directed to the
output link $j$ must be scaled down to satisfy the supply constraint $R_j$.

\item For each output $j$, compute the total flow directed to this output link:
\be
f_{1j} = \min\left\{\alpha_j\sum_{c=1}^C S_{1j}^c,\;
\min_{j'\neq j}\left\{(1-\eta_{j'j})\sum_{c=1}^C S_{1j}^c +
\eta_{j'j}\alpha_{j'}\sum_{c=1}^C S_{1j}^c
\right\}\right\}.
\label{eq_simo_directed_flow}
\ee
Here we pick the term that restricts the flow to output link $j$ the most:
the restriction may come from the output link $j$ itself,
or from other output links, denoted $j'$, that influence link $j$ through
mutual restriction coefficients $\eta_{j'j}$.

\item Assign input-output flows per commodity:
\be
f_{1j}^c = S_{1j}^c \frac{f_{1j}}{\sum_{c'=1}^C S_{1j}^{c'}}, \;\;\;
j=1,\dots,N.
\label{eq_simo_io_flow}
\ee
\end{enumerate}

\textbf{Example.} Suppose we have one input and three outputs, with parameters:

\begin{tabular}{l l l}
$C=1$ & $S=1000$ & \\
$R_1 = 100$ & $R_2 = 400$ & $R_3 = 300$ \\
$\beta_{11} = 0.2$ & $\beta_{12} = 0.5$ & $\beta_{13} = 0.3$ \\
$\eta_{12} = 0.2$ & $\eta_{21} = 1$ & $\eta_{31} = 0$ \\
$\eta_{13} = 0$ & $\eta_{23} = 0.5$ & $\eta_{32} = 0$
\end{tabular}

Our solution algorithm for finding the resulting flows proceeds as follows:
\begin{align*}
& 1. \qquad S_{11} = 200; \quad S_{12} = 500; \quad S_{13} = 300; \\
& 2. \qquad \alpha_1 = \frac{100}{200} = 0.5;
\quad \alpha_2 = \frac{400}{500} = 0.8;
\quad \alpha_3 = \frac{300}{300} = 1 \\
& 3. \\[-\baselineskip]
& \; \qquad  \begin{aligned}
f_{11} &= \min \bigg\{0.5 \cdot 200, \, \min \bigg\{ (1-1) 200 + 1 \cdot 0.8 \cdot 200, \, (1-0) 200 + 0 \cdot 1 \cdot 200 \bigg\} \bigg\} \\
&= \min \bigg\{100, \, \min \bigg\{160, \, 200 \bigg\} \bigg\} \\
\boldsymbol{f_{11}} &= \boldsymbol{100} \\
f_{12} &= \min \bigg\{0.8 \cdot 500, \, \min \bigg\{ (1-0.2) 500 + 0.2 \cdot 0.5 \cdot 500, \, (1-0) 500 + 0 \cdot 1 \cdot 500 \bigg\} \bigg\} \\
&= \min \bigg\{400, \, \min \bigg\{450, \, 500 \bigg\} \bigg\} \\
\boldsymbol{f_{12}} &= \boldsymbol{400} \\
f_{13} &= \min \bigg\{1 \cdot 300, \, \min \bigg\{ (1-0) 300 + 0 \cdot 0.2 \cdot 300, \, (1-0.5) 300 + 0.5 \cdot 0.8 \cdot 300 \bigg\} \bigg\} \\
&= \min \bigg\{300, \, \min \bigg\{300, \, 270 \bigg\} \bigg\} \\
\boldsymbol{f_{13}} &= \boldsymbol{270} \\
\end{aligned}
\end{align*}

Finally, we state the result of this Section as a theorem.
\begin{theo}
The SIMO input-output flow computation algorithm constructs the
unique solution of the maximization
problem~\eqref{mimo_objective}-\eqref{mimo_proportionality_constraint},~\eqref{simo_rfifo_constraint} with $M=1$.
\label{theo_simo_optimal}
\end{theo}


\subsection{MIMO Node with Relaxed FIFO Condition}\label{subsec_mimo_relaxed}
For a node with $M$ input and $N$ output links we are given
input demands $S_i^c$ and split ratios $\beta_{ij}$ specified
per commodity, input priorities $p_i\geq 0$,
output supply $R_j$ and mutual restriction matrices
$\{\eta_{jj'}\}^i$ ($\eta_{jj'}^i\in[0, 1]$, $\eta_{jj}^i=1$),
$i=1,\dots,M$, $j=1,\dots,N$ and $c=1,\dots,C$.
Oriented demand $S_{ij}^c$ and oriented priorities $p_{ij}$ are defined
in~\eqref{eq_oriented_demand}.

To formulate the optimization problem for the MIMO case
with the relaxed FIFO constraint, we extend
the SIMO relaxed FIFO constraint~\eqref{simo_rfifo_constraint}
to arbitrary $M$:
\be
\sum_{c=1}^C f_{ij}^c\leq
\left(1-\eta_{j'j}^i\right)\sum_{c=1}^C S_{ij}^c +
\eta_{j'j}
\left(\frac{\sum_{c=1}^C f_{ij'}^c}{\sum_{c=1}^C S_{ij'}^c}\right)
\left(\sum_{c=1}^C S_{ij}^c\right), \;
i=1,\dots,M, \;
j=1,\dots,N, \; \forall j'\neq j.
\label{mimo_rfifo_constraint}
\ee
If $\eta_{j'j}^i=1$ for all $i$, $j$ and $j'$, then
constraint~\eqref{mimo_rfifo_constraint} is equivalent
to~\eqref{mimo_fifo_constraint}.

The optimization 
problem~\eqref{mimo_objective}-\eqref{mimo_proportionality_constraint},~\eqref{mimo_priority_constraint},~\eqref{mimo_rfifo_constraint}
satisfies the requirements listed in the beginning
of Section~\ref{sec_node_model}.
The following algorithm for computing input-output flows $f_{ij}^c$ 
generalizes the MIMO algorithm from Section~\ref{subsec_mimo}
taking into account constraint~\eqref{mimo_rfifo_constraint}.
\begin{enumerate}
\item Initialize:
\begin{eqnarray*}
\Rt_j(0) & := & R; \\
U_j(0) & := & U_j; \\
\St_{ij}^c(0) & := & S_{ij}^c; \\
\St_{ij}(0) & := & \sum_{c=1}^C \St_{ij}^c(0);\\
k & := & 0; \\
& & i = 1,\dots,M, \;\;\; j = 1,\dots,N, \;\;\; c = 1,\dots,C.
\end{eqnarray*}

\item Define the set of output links that still need processing:
\[ V(k) = \left\{j: \; U_j(k) \neq \emptyset\right\}. \]
If $V(k) = \emptyset$, stop.

\item For each input link $i\in\bigcup_{j\in V(k)} U_j(k)$,
calculate input link priority $\pt_i(k)$ according to
expression~\eqref{eq_nonzero_priorities}.

\item For each output link $j\in V(k)$ and input links $i\in U_j(k)$,
calculate oriented priorities:
\be
\pt_{ij}(k) = \pt_i(k) \frac{\St_{ij}(k)}{
\sum_{j'\in V(k)} \St_{ij'}(k)}.
\label{eq_oriented_priorities_k}
\ee

\item For each $j\in V(k)$, compute factors $a_j(k)$
according to~\eqref{eq_aj_factors}
and find the most restrictive output link $j^\ast$:
\[
j^\ast = \arg\min_{j\in V(k)} a_j(k).
\]

\item Define the set of input links,
whose demand does not exceed the allocated supply:
\[
\Ut(k) = \left\{i\in U_{j^\ast}(k):\;
\left( \sum_{j\in V(k)} \St_{ij}(k) \right) \leq \pt_i(k)a_{j^\ast}(k)\right\}.
\]
\begin{itemize}
\item If $\Ut(k) \neq \emptyset$, then for all output links $j\in V(k)$ assign:
\begin{eqnarray}
f_{ij}^c & = & \St_{ij}^c(k), \;\;\; i\in\Ut(k),\;\; c=1,\dots,C; \label{eq_mimo_relaxed_freeflow} \\
\Rt_j(k+1) & = & \Rt_j(k) - \sum_{i\in\Ut(k)}\sum_{c=1}^C f_{ij}^c; \nonumber \\
U_j(k+1) & = & U_j(k) \setminus \Ut(k).  \nonumber
\end{eqnarray}
\item Else, for all input links $i\in U_{j^\ast}(k)$,
output links $j\in V(k)$ and commodities $c=1,\dots,C$, assign:
\begin{eqnarray}
f_{ij}^c & = & \St_{ij}^c(k)
\frac{\pt_{ij}(k)a_{j^\ast}(k)}{\sum_{c'=1}^C \St_{ij}^{c'}(k)},
\;\;\; j:~\eta_{j^\ast j}^i = 1;
\label{eq_mimo_relaxed_step6_1} \\
\St_{ij}^c(k+1) & = & \St_{ij}^c(k)\frac{\St_{ij}(k+1)}{\St_{ij}(k)},
\;\; j:~\eta_{j^\ast j}^i< 1, \;\; i\in U_j(k)\cap U_{j^\ast}(k);
\label{eq_mimo_relaxed_step6_2} \\
\mbox{ where} & & \nonumber \\
\St_{ij}(k+1) & = & \min\left\{\St_{ij}(k),\; \left(1-\eta_{j^\ast j}^i\right)
\sum_{c=1}^C S_{ij}^c + \eta_{j^\ast j}^{i} 
\frac{\sum_{c=1}^C f_{ij^\ast}^c}{\sum_{c=1}^C S_{ij^\ast}^c}
\sum_{c=1}^C S_{ij}^c \right\};
\label{eq_mimo_relaxed_step6_3} \\
\St_{ij}^c(k+1) & = & \St_{ij}^c(k), \;\;\;
i \not\in U_j(k)\cap U_{j^\ast}(k);\nonumber \\
\Rt_j(k+1) & = & \Rt_j(k) -
\sum_{i\in U_{j^\ast}(k):\eta_{j^\ast j}^i=1} \; \sum_{c=1}^C f_{ij}^c; \nonumber\\
U_j(k+1) & = & U_j(k) \setminus \left\{i\in U_{j^\ast}(k):\;
\eta_{j^\ast j}^i=1\right\}. \nonumber
\end{eqnarray}
Here, the assignment (\ref{eq_mimo_relaxed_step6_1}) reproduces that of
the original MIMO algorithm, (\ref{eq_mimo_step6}), and
the expression~\eqref{eq_mimo_relaxed_step6_3}
represents the selection of the minimum in
the SIMO algorithm, (\ref{eq_simo_directed_flow}),
where mutual restriction coefficients $\{\eta_{j^\ast j}\}^i$ are applied.
Equation~\eqref{eq_mimo_relaxed_step6_2} is an application of our
proportionality constraint for commodity flows,~\eqref{mimo_proportionality_constraint}.
\end{itemize}

\item Set $k:=k+1$, and return to step 2.
\end{enumerate}

This algorithm takes no more than $(M+N-2)$ iterations to complete.

{\bf Remark.}
An intuition for the composition in~\eqref{eq_mimo_relaxed_step6_3}
may be expressed as, ``is the currently-considered output $j^\ast$
the most restrictive output on the movement $(i,j)$ (RHS of the 
minimum), or has $j$ already had a greater restriction enforced upon it
(LHS of the minimum)?''

The following lemma states that in the case $M=1$, the MIMO algorithm
with relaxed FIFO condition produces exactly the same result as the 
SIMO algorithm with relaxed FIFO condition described in
Section~\ref{subsec_simo}.
\begin{lemma}
The SIMO algorithm with relaxed FIFO condition is a special case of the MIMO
algorithm with relaxed FIFO condition with $M=1$.
\label{lemma_mimo_simo}
\end{lemma}
\begin{proof}
In the case of $M=1$, factors $a_j(k)$, defined in~\eqref{eq_aj_factors},
reduce to:
\[
a_j(k) = \frac{R_j\sum_{c=1}^C \St_1^c(k)}{\sum_{c=1}^C \St_{1j}^c(k)}.
\]
Thus,
\[
j^\ast = \arg\min_{j\in V(k)}
\frac{R_j\sum_{c=1}^C \St_1^c(k)}{\sum_{c=1}^C \St_{1j}^c(k)}
= \arg\min_{j\in V(k)}\frac{R_j}{\sum_{c=1}^C \St_{1j}^c(k)},
\]
MIMO algorithm with $M=1$ goes into
iteration $k+1$ only if $\Ut(k)=\emptyset$, which is equivalent to
$\frac{R_{j^\ast}}{\sum_{c=1}^C \St_{1j^\ast}^c(k)}<1$.
Then the assignment of input-output
flows $f_{1j}^c$ for output links $j$, such that $\eta_{j^\ast j}^1=1$,
is given by~\eqref{eq_mimo_relaxed_step6_1} that translates to:
\[
f_{1j}^c=\alpha_{j^\ast}(k)\St_{1j}^c(k), \;\;\; c=1,\dots,C,
\]
where
\[
\alpha_{j^\ast}(k) = 
\frac{R_{j^\ast}}{\sum_{c=1}^C \St_{1j^\ast}^c} =
\min\left\{1, \; \frac{R_{j^\ast}}{\sum_{c=1}^C \St_{1j^\ast}^c(k)}
\right\}, 
\]
which is the same as~\eqref{eq_reduction_factor} for the output link $j^\ast$.
Obviously, $\sum_{c=1}^C f_{1j}^c=f_{1j}$
satisfies~\eqref{eq_simo_directed_flow} for all $j$, such that
$\eta_{j^\ast j}^1=1$, which always includes $j^\ast$.
For output links $j$, such that $\eta_{j^\ast j}<1$,
formula~\eqref{eq_mimo_relaxed_step6_3} achieves the same result
as~\eqref{eq_simo_directed_flow} through the iterative pairwise comparison.
\end{proof}

The next lemma states that the when FIFO condition is fully enforced,
that is, $\eta_{jj'}^i=1$ for all output link pairs $\{j,j'\}$ and
all input links $i$, the MIMO algorithm with relaxed FIFO condition
reduces to the original MIMO algorithm from Section~\ref{subsec_mimo}.
\begin{lemma}
The original MIMO-with-FIFO algorithm is a special case of the MIMO
algorithm with relaxed FIFO condition with mutual restriction coefficients
$\eta_{jj'}^i=1$, $i=1,\dots,M$, $j,j'=1,\dots,N$.
\end{lemma}
\begin{proof}
The proof follows from the fact that steps~6 of the MIMO-with-FIFO
and the MIMO-relaxed-FIFO algorithms are equivalent
when $\eta_{jj'}^i=1$ (compare expressions~\eqref{eq_mimo_step6}
and~\eqref{eq_mimo_relaxed_step6_1}).
\end{proof}

{\bf Remark.}
Note that in the case of multiple input links
mutual restriction coefficients are
specified per input link.
This can be justified by the following example.
In the node representing an intersection
with 2 input and 3 output links,
shown in Figure \ref{fig-intersection},
consider the influence of the output link 5
on the output link 4.
If vehicles enter links 4 and 5 from link 1, then
it is reasonable to assume that once link 5 is jammed
and cannot accept any vehicles, there is no flow
from 1 to 4 either. 
In other words, $\eta_{54}^1=1$.
On the other hand, if vehicles arrive from link 2,
blockage of the output link 5 may hinder, but not 
necessarily prevent traffic from flowing into the output
link 4, and so $\eta_{54}^2<1$.

\begin{figure}[htb]
\centering
\includegraphics[width=1.5in]{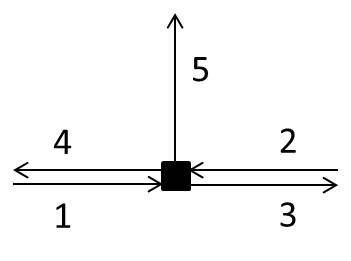}
\caption{Intersection node.}
\label{fig-intersection}
\end{figure}

\textbf{Remark.} As a rule of thumb, we suggest setting the restriction coefficient
$\eta_{jj'}^i$ equal to the portion of lanes serving movement $(i,j')$ that also serve
movement $(i,j)$. As an example, if movement $(i,j')$ is served by four lanes,
and of those lanes, one serves movement $(i,j)$, then a reasonable value for
$\eta_{jj'}^i$ is 0.25, as blockage of movement $(i,j)$ will block 25\% of
$(i,j')$'s lanes.

\textbf{Example.} Recall our earlier example for a 4-by-4 junction during our
discussion of the MIMO-with-FIFO node model in Section~\ref{subsec_mimo}.
Recall that our example supposed that the input links (see
Figure~\ref{fig-mimo_example} for the network topology) had capacity values
of $F_1=F_3=1000$, and
$F_2=F_3=2000$. Say that the links with capacity of 1000 are one-lane roads,
while those links with capacity of 2000 are two-lane roads, with left and
right turns permitted from one lane each, and the through movement permitted
from both lanes. Following our rule of thumb, this would lead to a set of
reduced mutual restriction coefficients:
\begin{alignat*}{2}
	\eta_{76}^4 &= \eta_{56}^4 = 0.5    \qquad&    \eta_{75}^4 &= \eta_{57}^4 = 0 \\
	\eta_{58}^2 &= \eta_{78}^2 = 0.5    &    \eta_{75}^2 &= \eta_{57}^2 = 0
\end{alignat*}
with all other restriction coefficients taking a value of 1. Let us examine
the effect of our FIFO relaxation on this example. Since $C=1$, we ignore
the index $c$ and perform some cancellations to clean up the notation.

\underline{$k=0:$}
This iteration will proceed exactly as before, with $\Ut(0) = \{1\}$
and flows from link 1 taking the demand value.

\underline{$k=1:$}
\begin{enumerate}
\setcounter{enumi}{1}
	\item $V(1) = \{5,6,7,8\}$
	\item No changes are made to priorities for any $i \in \bigcup_{j\in V(1)} U_j(1)$.
	\item No changes are made to oriented priorities.
	\item $a_5(1)=2.92; \quad a_6(1) = 1.83; \quad a_7(1) = 0.685; \quad a_8(0) = 0.723$ \\
		$a_{j^\ast}(1) = a_7(1) = 0.685$
	\item $\Ut(1) = \emptyset,$ as $\sum_j \St_{2j}(1)=2000 \nleq \pt_2(1) a_7(1) = 1370$ \\
	and $\sum_j \St_{4j}(1)=1700 \nleq \pt_4(1) a_7(1) = 1370$
	\begin{itemize}
		\item $\boldsymbol{f_{47} = \pt_{47}(1)a_7(1) = 644.5}$
		\item $\boldsymbol{f_{27} = \pt_{27}(1)a_7(1) = 205.5}$
		\item $\begin{aligned}[t]
			\St_{28}(2) &= (1- \eta_{78}^2) S_{28} + \eta_{78}^2 \times \sfrac{f_{27}}{S_{27}} \times S_{28} \\
						&= (1-0.5) \times 1600 + 0.5 \times \sfrac{205.5}{300} \times 1600 \\
						&= 1348
				\end{aligned}$
		\item $\begin{aligned}[t]
			\St_{46}(2) &= (1- \eta_{76}^4) S_{46} + \eta_{76}^4 \times \sfrac{f_{47}}{S_{47}} \times S_{46} \\
						&= (1-0.5) \times 800 + 0.5 \sfrac{644.5}{800} \times 800 \\
						&= 722.25
				\end{aligned}$
		\item $\Rt_7(2)=850-644.5-204.5=0$
		\item $U_5(2) = \{2,3,4\}; \quad U_6(2)=\{3,4\}; \quad U_7(2)=\emptyset; \quad U_8(2)=\{2,3\}$
	\end{itemize}
\end{enumerate}

\underline{$k=2:$}
\begin{enumerate}
\setcounter{enumi}{1}
	\item $V(2) = \{5,6,8\}$
	\item No changes are made to priorities $p_i$.
	\item Due to changes in two oriented demands, the new oriented priorities for input links 2 and 4 must be computed:
	\begin{itemize}
		\item $\pt_{25}(2) = 2000 \times \sfrac{100}{1448} = 138.12$
		\item $\pt_{28}(2) = 2000 \times \sfrac{1348}{1448} = 1861.9$
		\item $\pt_{45}(2) = 2000 \times \sfrac{100}{822.25} = 243.24$
		\item $\pt_{46}(2) = 2000 \times \sfrac{722.25}{822.25} = 1756.8$
	\end{itemize}
	\item The coefficients $a_j(2)$ are calculated with the new oriented priorities:
	\begin{itemize}
		\item $a_5(2) = \frac{\Rt_5(2)}{\pt_{25}(2)+\pt_{35}(2)+\pt_{45}(2)} = \frac{1000}{138.12+125+243.24} = 1.92$
		\item $a_6(2) = \frac{\Rt_6(2)}{+\pt_{36}(2)+\pt_{46}(2)} = \frac{1950}{125+1756.8} = 1.04$
		\item $a_8(2) = \frac{\Rt_8(2)}{+\pt_{28}(2)+\pt_{38}(2)} = \frac{1700}{1861.9+750} = 0.651 = a_{j^\ast}(2)$
	\end{itemize}
	\item $\Ut(2) = \emptyset,$ as $\sum_j \St_{2j}(2)=1448 \nleq \pt_2(2) a_8(2) = 1301$ \\
	 and $\sum_j \St_{3j}(2)=800 \nleq \pt_3(2) a_8(2) = 651,$
	 \begin{itemize}
		 \item $\boldsymbol{f_{28} = \pt_{28}(2) a_8(2) = 1861.9 \times 0.651 = 1211.75}$
		 \item $\boldsymbol{f_{38} = \pt_{38}(2) a_8(2) = 750 \times 0.651 = 488.25}$
		 \item $\eta_{85}^2=1,$ so $\boldsymbol{f_{25} = \pt_{25}(2) a_8(2) = 138.12 \times 0.651 = 89.916}$
		 \item $\eta_{85}^3=1,$ so $\boldsymbol{f_{35} = \pt_{35}(2) a_8(2) = 125 \times 0.651 = 81.375}$
		 \item $\eta_{86}^3=1,$ so $\boldsymbol{f_{36} = \pt_{36}(2) a_8(2) = 125 \times 0.651 = 81.375}$
		 \item $\Rt_5(3) = 1000 - 89.916 - 81.375 = 828.71; \quad \Rt_6(3) = 1950 - 81.375 = 1868.6$ \\
			 $\Rt_8(3) = 1700 - 1211.75 - 488.25 = 0$
		 \item $U_5(3) = \{4\}; \quad U_6(3) = \{4\}; \quad U_8(3) = \emptyset$
	 \end{itemize}
\end{enumerate}

\underline{$k=3:$}
\begin{enumerate}
\setcounter{enumi}{1}
	\item $V(3) = \{5,6\}$
	\item No changes are made to the priorities $p_4$, those of the only link $i$ remaining.
	\item The oriented demands of link 4 remain the same as calculated in the previous iteration.
	\item $a_5(3) = \sfrac{828.71}{258.39} = 3.21$ \\
		$a_6(3) = \sfrac{1861.6}{1741.6} = 1.073 = a_{j^\ast}(3)$
	\item $\Ut(3) = \{4\}$, as $\sum_j \St_{4j} = 774 \leq \pt_4(3) a_6(3) = 2146$ \\
		$\boldsymbol{ f_{45} = \St_{45}(3) = 100;} \quad \boldsymbol{ f_{46} = \St_{46}(3) = 722.25}$ \\
		$U_j(4) = \emptyset \quad \forall j$
\end{enumerate}

\underline{$k=4:$}
\begin{enumerate}
\setcounter{enumi}{1}
	\item $V(4) = \emptyset$, so the algorithm terminates. All flows have been found.
\end{enumerate}

Comparison of the resulting flows with those from the non-relaxed FIFO example
illustrates several differences caused by this relaxation. Unsurprisingly, the
two-lane links (2 and 4) that take advantage of the FIFO relaxation send greater
amounts of vehicles to links 5, 6 and 8. Previously, these flows had been cut
off by the filling of link 7's supply. Perhaps more interesting is the effect on
the flows of link 3. In the previous example, link 3 was able to send vehicles
equal to its demand, since there was enough leftover supply in link 8 after the
flow $f_{28}$ was constrained by FIFO. In our example with relaxed
FIFO, however, link 8 had more supply consumed by the higher-priority link 2, and
it was found that link 3 was unable to send its entire demand. Since we had
assumed that link 3 was a one-lane road and thus bounded by FIFO, this restricted
its other flows as well.

In both cases, link 3 was lower priority than link 2 and 
was disadvantaged in claiming link 8's supply, but qualitatively, permitting a
FIFO relaxation for the two-lane roads exacerbated the ``advantage'' they already
had over the one-lane roads (from having priority equal to capacity).

The main result of this Section can be stated as the following theorem.
\begin{theo}
  Given a set of input links $i \in 1,\dots,M$, output links $j \in 1,\dots,N$,
  commodities $c \in 1,\dots,C$, priorities $p_i$, split ratios $\beta_i$, and
  mutual restriction coefficients $\{ \eta_{j,j'}^i \}$, the 
  algorithm of Section~\ref{subsec_mimo_relaxed} obtains the
  unique solution of the optimization 
  problem~\eqref{mimo_objective}-\eqref{mimo_proportionality_constraint},~\eqref{mimo_priority_constraint},~\eqref{mimo_rfifo_constraint}.
\label{theo_mimo_relaxed_optimal}
\end{theo}
\begin{proof}
  The priority constraint~\eqref{mimo_priority_constraint}
  makes this optimization problem non-convex, except in the
  above-mentioned special case where $N=1$ and the priorities $p_i$
  are taken as proportional to the send functions. Recall our discussion in
  Section~\ref{subsec_mimo} that the priority constraint must be replaced
  entirely with the unrealistic constraint~\ref{mimo_demand_proportional}
  to obtain an LP. Our priority constraint increases the complexity of
  the problem significantly. In fact, we conjecture that
  in the MIMO case, both with and without relaxation of the FIFO constraint,
  there is no way to verify a solution faster than re-solving the
  problem~\eqref{mimo_objective}-\eqref{mimo_proportionality_constraint},~\eqref{mimo_priority_constraint},~\eqref{mimo_rfifo_constraint}. Thus, we will prove optimality
  by showing that as our algorithm proceeds through iterations, it
  constructs the unique optimal solution.
   
  We may decompose the problem into finding the $M \cdot N \cdot C$ interrelated quantities
  $\{f_{ij}^c\}$. The $C$ flows for each $(i,j)$ are further constrained by our
  commodity flow proportionality constraint~\eqref{mimo_proportionality_constraint};
  solving for one of $\{f_{ij}^1,\dots,f_{ij}^C\}$ also constrains them all. Our
  task then becomes finding optimal values for each of $M \cdot N$ subsets
  $\{f_{ij}^1,\dots,f_{ij}^C\}$. Our algorithm finds at least one of these $M \cdot N$
  subsets per iteration $k$. The assignments
  are done by either equation~\eqref{eq_mimo_relaxed_freeflow} or
  equation~\eqref{eq_mimo_relaxed_step6_1}.
  Over iterations, subsets are assigned to build up the unique optimum solution.
  We can show that each subset assigned is optimal; that is, at least one of the
  constraints is tight.
  
  Let us first consider formulae~\eqref{eq_mimo_relaxed_step6_2}-\eqref{eq_mimo_relaxed_step6_3},
  our implementation of the relaxed FIFO constraint. In step 5 of our algorithm,
  we identify a single output link as $j^\ast$. The argmin in this step picks
  out a single link as the most restrictive of all output links. By the relaxed
  FIFO construction, all $i\in U_{j^\ast}$ will feel a relaxed FIFO effect
  instigated by $j^\ast$ as the most restrictive link. For a generic
  $j \neq j^\ast$, equation~\eqref{eq_mimo_relaxed_step6_3} enforces the
  relaxed FIFO constraint by decaying the oriented demand $\St_{ij}$. In fact, this
  oriented demand acts as a proxy for the restricted FIFO constraint. Since a $f_{ij}^c$ that
  is restricted by relaxed FIFO will, by construction, never obtain $f_{ij}^c=S_{ij}^c$,
  the running quantity $S_{ij}^c(k)$ represents an ``effective'' oriented demand
  after application of restricted FIFO. If it turns out that a flow is found
  $f_{ij}^c=\St_{ij}^c(k) \leq S_{ij}^c$ for some $k$, then this flow value has an
  active constraint imposed by relaxed FIFO.
    
  Now consider flows assigned by equation~\eqref{eq_mimo_relaxed_step6_1}.
  In the special case where $j=j^\ast$, we have
    \begin{align*}
      \sum_{c=1}^C \sum_{i \in U_{j^\ast}(k)} f_{ij^\ast}^c &= \sum_{c=1}^C \sum_{i \in U_{j^\ast}(k)} \St_{ij^\ast}(k) \frac{\pt_{ij^\ast}(k) a_{j^\ast}(k)}{\sum_{c'=1}^C \St_{ij^\ast}^{c'}(k)} \\
      &= \sum_{i \in U_{j^\ast}(k)} \pt_{ij^\ast}(k) a_{j^\ast}(k) \\
      &= \left( \sum_{i \in U_{j^\ast}(k)} \pt_{ij^\ast}(k) \right) \frac{\Rt_{j^\ast}(k)}{\sum_{i \in U_{j^\ast}(k)} \pt_{ij^\ast}(k)} \\
      &= \Rt_{j^\ast}(k),
    \end{align*}
  so all the flows into link $j^\ast$ take up all available supply. These
  flows are thus constrained by the supply constraint,~\eqref{mimo_supply_constraint}.
  
  For flows assigned
  by equation~\eqref{eq_mimo_relaxed_step6_1} in the case of $j\neq j^\ast$,
  we have $i \in U_{j^\ast}$, and a restriction coefficient $\eta_{j^\ast j}^i=1$
  that implies a strict FIFO constraint. In other words, we need
  \begin{align*}
	  \frac{\sum_{c=1}^C f_{ij}^c}{\sum_{c=1}^C \left( f_{ij}^c + f_{ij^\ast}^c \right)} = 
	  \frac{\sum_{c=1}^C \St_{ij}^c(k)}{\sum_{c=1}^C \left( \St_{ij}^c(k) + \St_{ij^\ast}^c(k) \right)}.
  \end{align*}
  Examining the denominator of the LHS, we can write
  \begin{align*}
	  \sum_{c=1}^C f_{ij}^c + f_{ij^\ast}^c &= \sum_{c=1}^C \left( \pt_{ij}^c(k) + \pt_{ij^\ast}^c(k) \right) a_{j^\ast}(k) \\
	  &= \sum_{c=1}^C \frac{\pt_i(k)}{\sum_{j \in V(k)} \St_{ij}^c(k)} \left( \St_{ij}(k) + \St_{ij^\ast}(k) \right) a_{j^\ast}(k).
  \end{align*}
  Straightforward algebra then gives
  \begin{align*}
	  \frac{\sum_{c=1}^C f_{ij}^c}{\sum_{c=1}^C \left( f_{ij}^c + f_{ij^\ast}^c \right)}
	  &= \sum_{c=1}^C \frac{\pt_i(k)}{\pt_i(k)} \times \frac{\St_{ij}^c(k)}{\St_{ij}^c(k) + \St_{ij^\ast}^c(k)}
		  \times \frac{\sum_{j \in V(k)} \St_{ij}^c(k)}{\sum_{j \in V(k)} \St_{ij}^c(k)} \times \frac{a_{j^\ast}(k)}{a_{j^\ast}(k)} \\
      &= \frac{\sum_{c=1}^C \St_{ij}^c(k)}{\sum_{c=1}^C \left( \St_{ij}^c(k) + \St_{ij^\ast}^c(k) \right)}.
  \end{align*}
  as originally required. Thus, flows $\{f_{ij}^c\}, j\neq j^\ast$ assigned in
  equation~\eqref{eq_mimo_relaxed_step6_1} are constrained by the strict FIFO
  constraint,~\eqref{mimo_rfifo_constraint} with $\eta_{j^\ast j}^i=1$.
  
  For flows assigned by equation~\eqref{eq_mimo_relaxed_freeflow}, we have
  for each flow $i$, $f_{ij}^c = \St_{ij}^c(k)$, so these flows are constrained
  by~\eqref{mimo_demand_constraint}.
  
  Uniqueness of the solution is equivalent to showing that we must assign
  flows in the order described. That is, that flows for the most restrictive
  link $j^\ast$ should be assigned at iteration $k$. It turns out that setting
  the flows of any other $j\neq j^\ast$ may produce solutions that violate the
  problem constraints. At iteration $k$, let $j'$
  be some other constrained outgoing link, with $a_{j^\ast}(k) < a_{j'}(k)$,
  and let $i'$ be any input link in $U_{j^\ast}(k) \cap U_{j'}(k)$.
  Proceeding through our algorithm, we will obtain an upper bound on the sum of
  flows from $i'$ to $j'$ through the restriction coefficient:
  \begin{align*}
  \sum_{c=1}^C f_{i'j'}^c &\leq \min \left\{ \St_{i'j'}(k+1), \pt_{i'}(k+1) a_{j'}(k+1) \right\} \\
  &\leq \St_{i'j'}(k+1) \\
  &\leq (1-\eta_{j^\ast j'}^{i'}) \sum_{c=1}^C S_{i'j'}^c + 
  \eta_{j^\ast j'}^{i'} \frac{ \sum_{c=1}^C f_{i'j^\ast}^c}{\sum_{c=1}^C S_{i'j^\ast}^c} \sum_{c=1}^C S_{i'j'}^c \triangleq \overline{f_{i'j'}}.
  \end{align*}
  
  Suppose instead that we had decided to set flows into $j'$ at iteration $k$.
  Then we would have
  \begin{align}
  \sum_{c=1}^C f_{i'j'}^c &= \pt_{i'j'}(k) a_{j^\ast}(k)
  = \frac{\pt_{i'j'}(k)}{\sum_{i \in U_{j'}(k)} \pt_{i'j'}(k)} \Rt_{j'}(k).
  \label{eq_flow_max_proof}
  \end{align}
  
  Note that~\eqref{eq_flow_max_proof} does not make use of the restriction
  coefficient $\eta_{j^\ast j'}^{i'}$, and therefore does not take (relaxed)
  FIFO into consideration. In fact, for a suitably large $\eta_{j^\ast j'}^{i'}$
  and/or small ratio $\frac{ \sum_{c=1}^C f_{i'j^\ast}^c}{\sum_{c=1}^C S_{i'j^\ast}^c}$, we can obtain
  \begin{align*}
  \frac{\pt_{i'j'}(k)}{\sum_{i \in U_{j'}(k)} \pt_{i'j'}(k)} \Rt_{j'}(k) &>
  \left( \eta_{j^\ast j'}^{i'}
  \frac{\sum_{c=1}^C f_{i'j^\ast}^c}{\sum_{c=1}^C S_{i'j^\ast}^c}
  + 1- \eta_{j^\ast j'}^{i'} \right) \sum_{c=1}^C S_{i'j'}^c \\
  \sum_{c=1}^C f_{i'j'}^c &> \overline{f_{i'j'}},
  \end{align*}
  breaking our upper bound imposed by (relaxed) FIFO~\eqref{mimo_rfifo_constraint}.
\end{proof}

\begin{cor}
Theorems~\ref{theo_miso_optimal},~\ref{theo_mimo_optimal},
and~\ref{theo_simo_optimal} follow from the
Theorem~\ref{theo_mimo_relaxed_optimal} as special cases.
\end{cor}

\textbf{Remark.} It was discussed above that the Tamp\'{e}re et al.
node models are special cases of this node model (with all $\eta^i_{jj'}=1$
and certain choices of $p_i$).
The above proof then also applies to the similar theorem for that model.

\section{Node with Undefined Split Ratios: Traffic Assignment}\label{sec_sr_assignment}
\newcommand{\obeta}{\overline{\beta}}
\newcommand{\tbeta}{\tilde{\beta}}
In this section we consider a MIMO node with $M$ input links, $N$ output links
and $C$ commodities, where some of the split ratios $\beta_{ij}^c$
are not defined a priori and must be computed as functions
of the input demand $S_i^c$, priorities $p_i$ and the output supply $R_j$,
$i=1,\dots,M$, $j=1,\dots,N$ and $c=1,\dots,C$.
This may occur if the modeler chooses to let drivers at a node select
a route to their destination in response to changing conditions;
a specific example of such a node is given in an example
describing an interchange between a freeway
and a managed lane to be discussed after presentation of the
algorithm.
Here we present the algorithm for computing undefined split ratios
based on the following definitions and assumptions:
\begin{itemize}
\item Define the set of commodity movements, for which split ratios
are known, $\BB=\left\{\left\{i,j,c\right\}: \; \beta_{ij}^c \in [0,1]\right\}$,
and the set of commodity movements, for which split ratios are to be
computed, $\OBB=\left\{\left\{i,j,c\right\}: \; \beta_{ij}^c \mbox{ are unknown}\right\}$.

\item For a given input link $i$ and commodity $c$ such that $S_i^c=0$,
assume that all split ratios are known: $\{i,j,c\}\in\BB$.\footnote{If split
ratios were undefined in this case, they could be assigned arbitrarily.}

\item Define the set of output links, for which there exist unknown
split ratios, $V=\left\{j: \; \exists \left\{i,j,c\right\}\in\OBB\right\}$.

\item Assuming that for a given input link $i$ and commodity $c$
split ratios must sum up to 1, define the unassigned portion
of flow $\obeta_i^c=1-\sum_{j:\{i,j,c\}\in\BB}\beta_{ij}^c$.

\item For a given input link $i$ and commodity $c$ such that there exist
$\{i,j,c\}\in\OBB$, assume $\obeta_i^c>0$, otherwise the
undefined split ratios can be trivially set to 0.

\item For every output link $j\in V$, define the set of input links
that have an unassigned demand portion directed toward this output link,
$U_j=\left\{i: \; \exists\left\{i,j,c\right\}\in\OBB\right\}$.

\item For a given input link $i$ and commodity $c$ define the set
of output links, split ratios for which are to be computed,
$V_i^c = \left\{j: \; \exists i\in U_j\right\}$,
and assume that if nonempty, this set contains at least two elements,
otherwise a single split ratio can be trivially set to $\obeta_i^c$.

\item Assume that input link priorities are nonnegative, $p_i\geq 0$,
	$i=1,\dots,M$, and $\sum_{i=1}^M p_i = 1$. Note that, although in
	section~\ref{sec_node_model} we did not require the input priorities
	to sum to one, in this section we assume this normalization is done
	to simplify the notation.

\item Define the set of input links with zero priority:
$U_{zp} = \left\{i:\;p_i=0\right\}$.
To avoid dealing with zero input priorities, perform regularization:
\be
\pt_i = p_i\left(1-\frac{|U_{zp}|}{M}\right)+ \frac{1}{M}\frac{|U_{zp}|}{M}=
p_i\frac{M-|U_{zp}|}{M} + \frac{|U_{zp}|}{M^2},
\label{priority_regularization}
\ee
where $|U_{zp}|$ denotes the number of elements in set $U_{zp}$.
Expression~\eqref{priority_regularization} implies that the regularized
input priority $\pt_i$ consists of two parts:
(1) the original input priority $p_i$ normalized to the portion of 
input links with nonzero priorities; and
(2) uniform distribution among $M$ input links, $\frac{1}{M}$, 
normalized to the portion of input links with zero priorities.

Note that $\pt_i\geq 0$, $i=1,\dots,M$, and $\sum_{i=1}^M \pt_i = 1$.
\end{itemize}

The algorithm for distributing $\obeta_i^c$ among the commodity movements
in $\OBB$, that is assigning values to the a priori unknown split ratios,
aims at maintaining output links as uniformly occupied as possible.
It is described next.

\begin{enumerate}
\item Initialize:
\begin{eqnarray*}
\tbeta_{ij}^c(0) & := & \left\{\begin{array}{ll}
\beta_{ij}^c, & \mbox{ if } \{i,j,c\}\in\BB,\\
0, & \mbox{ otherwise};\end{array}\right. \\
\obeta_i^c(0) & := & \obeta_i^c; \\
\Ut_j(0) & = &  U_j; \\
\Vt(0) & = & V; \\
k & := & 0,
\end{eqnarray*}
Here $\Ut_j(k)$ is the remaining set of input links with some unassigned demand,
which may be directed to output link $j$; and
$\Vt(k)$ is the remaining set of output links, to which the still unassigned
demand may be directed.

\item If $\Vt(k)=\emptyset$, stop.
The sought-for split ratios are $\left\{\tbeta_{ij}^c(k)\right\}$,
$i=1,\dots,M$, $j=1,\dots,N$, $c=1,\dots,C$.

\item Calculate the remaining unallocated demand:
\[
\So_i^c(k) = \obeta_i^c(k) S_i^c, \;\;\; i=1,\dots,M, \;\; c=1,\dots,C.
\]

\item For all input-output link pairs
calculate oriented demand analogously to (\ref{eq_oriented_demand}):
\[ \St_{ij}^c(k) = \tbeta_{ij}^c(k) S_i^c. \]

\item For all input-output link pairs calculate oriented priorities:
\begin{eqnarray}
\pt_{ij}(k) & = & \pt_i\frac{\sum_{c=1}^C \gamma_{ij}^c S_i^c}{
\sum_{c=1}^C S_i^c}
\label{eq_oriented_priorities_undefined_sr1} \\
\mbox{with} & & \nonumber \\
\gamma_{ij}^c(k) & = & \left\{\begin{array}{ll}
\beta_{ij}^c, & \mbox{ if split ratio is defined a priori: }
\{i,j,c\}\in\BB, \\
\tbeta_{ij}^c(k) + \frac{\obeta_i^c(k)}{|V_i^c|}, &
\mbox{ otherwise},\end{array}\right.
\label{eq_oriented_priorities_undefined_sr2}
\end{eqnarray}
where $|V_i^c|$ denotes the number of elements in the set $V_i^c$.
Comparing the
expression~\eqref{eq_oriented_priorities_undefined_sr1}-\eqref{eq_oriented_priorities_undefined_sr2}
with~\eqref{eq_oriented_priorities}, one can see that 
split ratios $\tbeta_{ij}^c(k)$, which are not fully defined yet,
are complemented with a fraction of $\obeta_i^c(k)$ inversely proportional
to the number of output links among which the flow of commodity $c$
from input link $i$ can be distributed.

Note that in this step we user \emph{regularized} priorities $\pt_i$
as opposed to the original $p_i$, $i=1,\dots,M$.
This is done to ensure that inputs with $p_i=0$ are not ignored
in the split ratio assignment.

\item Find the largest oriented demand-supply ratio:
\[
\mu^+(k) = \max_j\max_i
\frac{\sum_{c=1}^C \St_{ij}^c(k)}{\pt_{ij}(k)R_j}\sum_{i\in U_j}\pt_{ij}(k).
\]

\item Define the set of all output links, where the minimum of
the oriented demand-supply ratio is achieved:
\[
Y(k) = \arg\min_{j\in\Vt(k)}\min_{i\in\Ut_j(k)}
\frac{\sum_{c=1}^C \St_{ij}^c(k)}{\pt_{ij}(k)R_j}
\sum_{i\in U_j}\pt_{ij}(k),
\]
and from this set pick the output link $j^-$ with the smallest
output demand-supply ratio is minimal (when there are multiple
minimizing output links, any of the minimizing output links
may be chosen as $j^-$):
\[
j^- = \arg\min_{j\in Y(k)}
\frac{\sum_{i=1}^M\sum_{c=1}^C\St_{ij}^c(k)}{R_j}.
\]

\item Define the set of all input links, where the minimum of
the oriented demand-supply ratio for the output link $j^-$ is achieved:
\[
W_{j^-}(k) = \arg\min_{i\in\Ut_{j^-}(k)}
\frac{\sum_{c=1}^C \St_{ij^-}^c(k)}{\pt_{ij^-}(k)R_{j^-}}
\sum_{i\in U_{j^-}}\pt_{ij^-}(k),
\]
and from this set pick the input link $i^-$ and commodity $c^-$
with the smallest remaining unallocated demand:
\[
\{i^-, c^-\} = \arg\min_{\begin{array}{c}
i\in W_{j^-}(k),\\
c:\obeta_{i^-}^c(k)>0\end{array}} \So_i^c(k).
\]

\item Define the smallest oriented demand-supply ratio:
\[
\mu^-(k) = 
\frac{\sum_{c=1}^C \St_{i^-j^-}^c(k)}{\pt_{i^-j^-}(k)R_{j^-}}
\sum_{i\in U_{j^-}}\pt_{ij-}(k).
\]
\begin{itemize}
\item If $\mu^-(k) = \mu^+(k)$, which means that the oriented demand
is perfectly balanced among the output links, 
distribute the unassigned demand proportionally
to the allocated supply:
\begin{eqnarray}
\tbeta_{i^-j}^{c^-}(k+1) & = & \tbeta_{i^-j}^{c^-}(k) + 
\frac{\pt_{i^-j}(k)R_j}{\sum_{j'\in\Vt_{i^-}^{c^-}(k)}\pt_{i^-j'}(k)R_{j'}}
\obeta_{i^-}^{c^-}(k), \;\;\; j\in \Vt_{i^-}^{c^-}(k);
\label{eq_sr_assign_1}\\
\tbeta_{ij}^{c}(k+1) & = & \tbeta_{ij}^{c}(k) \mbox{ for } 
\{i,j,c\}\neq\{i^-,j,c^-\}; \label{eq_sr_assign_2}\\
\obeta_{i^-}^{c^-}(k+1) & = & 0; \nonumber\\
\obeta_i^c(k+1) & = & \obeta_i^c(k) \mbox{ for }
\{i,c\}\neq\{i^-,c^-\}.\nonumber
\end{eqnarray}

\item Else, assign:
\begin{eqnarray}
\Delta\tbeta_{i^-j^-}^{c^-}(k) & = & \min\left\{\obeta_{i^-}^{c^-}(k),\;\;
\left(\frac{\mu^+(k)\pt_{i^-j^-}(k)R_{j^-}}{
\So_{i^-}^{c^-}(k) \sum_{i\in U_{j^-}}\pt_{ij^-}(k)} -
\frac{\sum_{c=1}^C\St_{i^-j^-}^c(k)}{\So_{i^-}^{c^-}(k)}\right)\right\};
\label{eq_sr_assign_3}\\
\tbeta_{i^-j^-}^{c^-}(k+1) & = & \tbeta_{i^-j^-}^{c^-}(k) + 
\Delta\tbeta_{i^-j^-}^{c^-}(k);
\label{eq_sr_assign_4}\\
\tbeta_{ij}^{c}(k+1) & = & \tbeta_{ij}^{c}(k) \mbox{ for }
\{i,j,c\}\neq\{i^-,j^-,c^-\}; \label{eq_sr_assign_5}\\
\obeta_{i^-}^{c^-}(k+1) & = & \obeta_{i^-}^{c^-}(k) -
\Delta\tbeta_{i^-j^-}^{c^-}(k); \nonumber \\
\obeta_i^c(k+1) & = & \obeta_i^c(k) \mbox{ for } \{i,c\}\neq\{i^-,c^-\}.
\nonumber
\end{eqnarray}
\end{itemize}

\item Update sets $\Ut_j(k)$ and $\Vt(k)$:
\begin{eqnarray*}
\Ut_j(k+1) & = & \Ut_j(k) \setminus \left\{i^-: \;
\obeta_{i^-}^c(k+1) = 0, \; c=1,\dots,C\right\}, \;\;\; j\in \Vt(k);\\
\Vt(k+1) & = & \Vt(k) \setminus \left\{j: \; \Ut_j(k+1)=\emptyset\right\}.
\end{eqnarray*}

\item Set $k := k+1$ and return to step 2.
\end{enumerate}


\begin{figure}[h]
\centering
\includegraphics[width=1.5in]{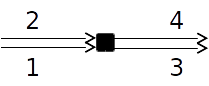}
\caption{Merge-diverge intersection example.}
\label{fig-sr_example}
\end{figure}

\textbf{Example.} Consider the node presented in
Figure~\ref{fig-sr_example}. The node represents a junction between
two parallel links, or a merge-diverge. Nodes such of these may be
used to model locations where drivers can choose to shift between
two parallel roads or lane groups; the canonical example is an
access point for a managed lane facility adjacent to a freeway,
where drivers may choose to enter or exit the managed lane. In this
example, we will say that links 1 and 3 represent a freeway's general
purpose (GP) lanes, and links 2 and 4 represent a high-occupancy vehicle
(HOV) lane, open only to a select group of vehicles. The Low-Occupancy
Vehicles (LOVs) and HOVs will be modeled as separate commodities, with
notations of $L$ for LOVs and $H$ for HOVs.

Suppose we have input parameters of:

\begin{tabular}{l l l}
 $\beta_{14}^L = 0$ & $\beta_{13}^L = 1$ & \\
 $\beta_{14}^H = ?$ & $\beta_{13}^H = ?$ & $p_1=\sfrac{3}{4}$ \\
 $\beta_{23}^H = ?$ & $\beta_{24}^H = ?$ & $p_2=\sfrac{1}{4}$ \\
 $S_1^L = 500$ & $S_1^H = 100$ & $R_3 = 600$ \\
 $S_2^H = 50$ & & $R_4 = 200$ \\
\end{tabular}

In words, the LOVs are not allowed to enter the HOV lane; HOVs are allowed,
but individual vehicles may or may not choose to do so: both links can
bring them to their eventual downstream destination so the choice of link
to take will be made in response to immediate local congestion conditions.

Let us see how algorithm would assign split ratios for the HOVs.

\underline{$k=0:$}
\begin{enumerate}
	\item $\tbeta_{13}^H(0) = \tbeta_{14}^H(0) = \tbeta_{23}^H(0) = \tbeta_{24}^H(0) = 0$ \\
		$\obeta_3^H(0) = \obeta_4^H(0) = 1$ \\
		$\Ut_3(0) = \Ut_4(0) = \{1,2\}; \qquad \Vt(0) = \{3,4\}$
	\item $\Vt(0) = \{3,4\} \neq \emptyset$, the algorithm proceeds.
	\item $\So^H_1(0) = 100, \qquad \So^H_2(0) = 50, \qquad \So_1^L(0) = 0$
	\item $\St_{13}^L(0) = 500; \qquad \St_{13}^H(0) = \St_{14}^H(0) = \St_{23}^H(0) = \St_{24}^H(0) = 0$
	\item $\gamma_{13}^H = \gamma_{14}^H = \sfrac{1}{2}; \qquad \gamma_{23}^H = \gamma_{24}^H = \sfrac{1}{2}; \qquad \gamma_{13}^L = 1$
	
		\vspace{5pt}
		$\begin{aligned}[t]
			\pt_{13}(0) &= \frac{\pt_1 \left( \gamma_{13}^H S_1^H + \gamma_{13}^L S_1^L \right)}{S_1^H + S_1^L}
						= \frac{ \sfrac{3}{4} \left( \sfrac{1}{2} \times 100 + 1 \times 500\right)}{100 + 500}
						= 0.6875 \\
			\pt_{14}(0) &= \frac{ \sfrac{3}{4} \left( \sfrac{1}{2} \times 100 + 0 \times 500\right)}{100 + 500}
						= 0.0625 \\
			\pt_{23}(0) &= \frac{\sfrac{1}{4} \left( \sfrac{1}{2} \times 50 \right)}{50} = 0.125 \\
			\pt_{24}(0) &= 0.125
		\end{aligned}$
	\item $\begin{aligned}[t]
			\mu^+(0) = \frac{\St_{13}^L(0) + \St_{13}^H(0)}{\pt_{13}(0) R_3} \sum_{i \in U_3} \pt_{i3}(0) = 0.9848
		\end{aligned}$, for all other pairs the ratio equals 0.
	\item $Y(0) = \{3,4\};$ \\
		$j^- = 4$
	\item $W_4(0) = \{1,2\}$, as there are no oriented demands to link 4 as of yet. \\
		$(i^-,c^-) = (2,H)$
	\item $\mu^-(0) = 0$
	
		\vspace{5pt}
		$\begin{aligned}[t]
			\Delta \tbeta_{i^- j^-}^{c^-}(0) &= \Delta \tbeta_{24}^H(0)
				= \min \left\{ 1, \left( \frac{ 0.9848 \times 0.125 \times 200}{50 \times 0.1875} - \frac{0}{50} \right) \right\} = \min \left\{1, 2.616 \right\}
				= 1 \\
			\tbeta_{24}^H(1) &= 0 + 1 = 1 \\
			\obeta_{2}^H(1) &= 1 - 0 = 0
		\end{aligned}$
	\item $\Ut_3(1) = \Ut_4(1) = \{1\}; \qquad \Vt(1) = \{3,4\}$
\end{enumerate}

\underline{$k=1:$}
\begin{enumerate}
\setcounter{enumi}{1}
	\item $\Vt(1) = \{3,4\}$
	\item $\So_1^H(1) = 100, \qquad \So^H_2(1) = 0, \qquad \So_1^L(1) = 0$
	\item $\St_{13}^L(0) = 500; \qquad \St_{13}^H(1) = \St_{14}^H(1) = \St_{23}^H(1) = 0; \qquad \St_{24}^H(1) = 50$
	\item $\gamma_{13}^H = \gamma_{14}^H = \sfrac{1}{2}; \qquad \gamma_{23}^H = 0; \qquad \gamma_{24}^H = 1; \qquad \gamma_{13}^L = 1$ \\
		$\pt_{13}(1) = 0.6875; \qquad \pt_{14}(1) = 0.0625;$ \\
		$\pt_{23}(1) = 0; \qquad \pt_{24}(1) = 0.25$
	\item 
	
		\vspace{5pt} $\begin{aligned}[t]
				&\frac{\St_{13}^L(1) + \St_{13}^H(1)}{\pt_{13}(1) R_3} \sum_{i \in U_3} \pt_{i3}(1) = 0.8333 = \mu^+(1) \\
				&\frac{\St_{24}^L(1) + \St_{24}^H(1)}{\pt_{24}(1) R_4} \sum_{i \in U_4} \pt_{i4}(1) = 0.3125
		\end{aligned}$
		
		\vspace{5pt} For the two other pairs the ratio equals 0.
	\item $Y(1) = \{4\};$ \\
		$j^- = 4$
	\item $W_4(1) = \{1\}$ \\
		$(i^- c^-) = (1, H)$
	\item $\mu^-(1) = 0$
		
			\vspace{5pt}
			$\begin{aligned}[t]
				\Delta \tbeta_{i^- j^-}^{c^-}(1) &= \Delta \tbeta_{14}^H(1)
					= \min \left\{ 1, \left( \frac{ 0.8333 \times 0.0625 \times 200}{100 \times 0.3125} - \frac{0}{100} \right) \right\} = \min \left\{1, \sfrac{1}{3} \right\}
					= \sfrac{1}{3} \\
				\tbeta_{14}^H(2) &= 0 + \sfrac{1}{3} = \sfrac{1}{3} \\
				\obeta_{2}^H(2) &= 1 - \sfrac{1}{3} = \sfrac{2}{3}
			\end{aligned}$
		\item $\Ut_3(2) = \Ut_4(2) = \{1\}; \qquad \Vt(2) = \{3,4\}$
\end{enumerate}

\underline{$k=2:$}
\begin{enumerate}
\setcounter{enumi}{1}
	\item $\Vt(2) = \{3,4\}$
	\item $\So_1^H(2) = 66.67, \qquad \So^H_2(2) = 0, \qquad \So_1^L(2) = 0$
	\item $\St_{13}^L(2) = 500; \qquad \St_{24}^H(2) = 50; \qquad \St_{14}^H(2) = 33.33; \qquad \St_{13}^H(2) = \St_{23}^H(2) = 0$
	\item $\gamma_{13}^H = \sfrac{1}{3}; \qquad \gamma_{14}^H = \sfrac{2}{3}; \qquad \gamma_{23}^H = 0; \qquad \gamma_{24}^H = 1; \qquad \gamma_{13}^L = 1$ \\
		$\pt_{13}(2) = \sfrac{2}{3}; \qquad \pt_{14}(2) = \sfrac{1}{12};$ \\
		$\pt_{23}(2) = 0; \qquad \pt_{24}(2) = 0.25$
	\item 
	
		\vspace{5pt} $\begin{aligned}[t]
				&\frac{\St_{13}^L(2) + \St_{13}^H(2)}{\pt_{13}(2) R_3} \sum_{i \in U_3} \pt_{i3}(2) = 0.8333 = \mu^+(2) \\
				&\frac{\St_{14}^L(2) + \St_{14}^H(2)}{\pt_{14}(2) R_4} \sum_{i \in U_4} \pt_{i4}(2) = 0.6667 \\
				&\frac{\St_{23}^L(2) + \St_{23}^H(2)}{\pt_{23}(2) R_3} \sum_{i \in U_3} \pt_{i3}(2) = 0 \\
				&\frac{\St_{24}^L(2) + \St_{24}^H(2)}{\pt_{24}(2) R_4} \sum_{i \in U_4} \pt_{i4}(2) = 0.3333
		\end{aligned}$
		
	\item $Y(2) = \{3\};$ \\
		$j^- = 3$
	\item $W_3(2) = \{1\}$ \\
		$(i^- c^-) = (1, H)$
	\item $\mu^-(2) = 0.8333 = \mu^+(2)$
		
			\vspace{5pt}
			$\begin{aligned}[t]
				\tbeta_{13}^H(3) &= 0 + \frac{ \sfrac{2}{3} \times 600}{
					\sfrac{2}{3} \times 600 + \sfrac{1}{12} \times 200} \times \sfrac{2}{3}
					= 0 + 0.96 \times \sfrac{2}{3} = 0.64 \\
				\tbeta_{23}^H(3) &= \sfrac{1}{3} + \frac{ \sfrac{1}{12} \times 200}{
					\sfrac{2}{3} \times 600 + \sfrac{1}{12} \times 200} \times \sfrac{2}{3}
					= \sfrac{1}{3} + 0.04 \times \sfrac{2}{3} = 0.36 \\
				\obeta_{1}^H(3) &= \sfrac{2}{3} - 0.64 - 0.0267 = 0
			\end{aligned}$
		\item $\Ut_3(3) = \Ut_4(3) = \emptyset; \qquad \Vt(3) = \emptyset$
\end{enumerate}

\underline{$k=3:$}
\begin{enumerate}
\setcounter{enumi}{1}
	\item $\Vt(3) = \emptyset$
\end{enumerate}
The algorithm terminates. The resulting split ratios are
$\boldsymbol{ \beta_{13}^H = 0.64, \quad
	\beta_{23}^H = 0.36, \quad \beta_{24}^H = 1,
	\quad \beta_{23}^H = 0}$.

\section{Conclusion}\label{sec_conclusion}
This paper introduced the first order macroscopic
traffic model, LNCTM,
for modeling multi-commodity traffic on road networks.
An example of the LNCTM application is the freeway network
with special lanes, where different vehicle classes have different permissions
to enter those special lanes. 

One of the features of the LNCTM is the fundamental diagram in the
shape of the ``inverse lambda'' that allows modeling of the capacity drop
and a hysteresis behavior of the traffic state in a link that goes
from free flow to congestion and back (Section~\ref{sec_lnctm}).

The cornerstone of the LNCTM is the model of a node
with multiple input and multiple output links, which computes
traffic flows per commodity
from input to output links and is governed by the input
demand, output supply, split ratios that define how incoming flows must
be distributed between output links, and input link priorities that
define the output supply allocation for incoming flows
(Sections~\ref{subsec_miso}-\ref{subsec_mimo}).
The proposed node model accepts arbitrary input link priorities,
and it is shown that the node model presented in~\citet{tampere11} is
the special case of the proposed node model when input link priorities are
proportional to input link capacities
(Sections~\ref{subsec_bliemer_tampere}).
We also analyzed the non-convex nature of the node model,
showing its difference
from the Bliemer node model \citep{bliemer07} that is formulated as an LP
(Section~\ref{subsec_bliemer_tampere}).

Next, we discussed the overly restrictive nature of the FIFO
rule on the output flows and further generalized the node model by
relaxing the FIFO condition.
Parameters, called \emph{mutual restriction coefficients} allow tuning
the node model from no FIFO when they are set to~0 to full FIFO when they are
set to~1 (Section~\ref{subsec_simo}).
Then, it was shown that this generalized node model produces unique
input-output flow allocation that maximizes the total output flow
of the node, subject to constraints imposed by input demand, output supply,
split ratios, mutual restriction coefficients and input link priorities
(Section~\ref{subsec_mimo_relaxed}).

In freeways with activated HOV lane LOVs are confined to the GP lane,
whereas HOVs may choose between GP and HOV lanes based on traffic conditions,
and so split ratios for commodity representing HOV traffic may not be
a priori defined, but must be computed at runtime. 
Therefore, we introduced the algorithm for local traffic assignment
that computes a priori undefined split ratios at nodes based on
input demand, output supply and input link priorities
(Section~\ref{sec_sr_assignment}).

The above mentioned results are presented in the form of constructive
computational algorithms that are readily implementable in
traffic simulation software.

\section*{Acknowledgements}\label{sec_acknowledgement}
We would like to express great appreciation to our colleagues
Elena Dorogush and Ajith Muralidharan for sharing ideas,
Ramtin Pedarsani, Brian Phegley and Pravin Varaiya for their
critical reading and their help in clarifying some theoretical issues.

This research was funded by the California Department of Transportation.

\appendix
\section{Notation}\label{app_notation}
\begin{longtable}[l]{l p{4.5in} l}
	Symbol & Definition & Used in section(s) \\
	\hline
	$t$ & Timestep index & \ref{sec_lnctm}\\
	$T$ & Final timestep & \ref{sec_lnctm} \\
	$k$ & Iteration index & \ref{sec_node_model},\ref{sec_sr_assignment}\\
	$l$ & Generic link index & \ref{sec_lnctm}\\
	$\mathcal{L}$ & Set of all links & \ref{sec_lnctm} \\
	$\nu$ & Generic node index & \ref{sec_lnctm} \\
	$\mathcal{N}$ & Set of all nodes & \ref{sec_lnctm} \\
	$i$ & Index of links entering a node & \ref{sec_node_model},\ref{sec_sr_assignment}\\
	$M$ & Number of links entering a node & \ref{sec_node_model},\ref{sec_sr_assignment}\\
	$j$ & Index of links exiting a node & \ref{sec_node_model},\ref{sec_sr_assignment}\\
	$N$ & Number of links exiting a node & \ref{sec_node_model},\ref{sec_sr_assignment}\\
	$c$ & Vehicle commodity index & \ref{sec_lnctm},\ref{sec_node_model},\ref{sec_sr_assignment}\\
	$C$ & Number of vehicle commodities & \ref{sec_lnctm},\ref{sec_node_model},\ref{sec_sr_assignment}\\
	$n_l^c$ & Density of vehicle commodity $c$ in link $l$ & \ref{sec_lnctm}\\
	$f_{ij}^c$ & Flow of vehicle commodity $c$ from link $i$ to link $j$ & \ref{sec_lnctm},\ref{sec_node_model} \\
	$F_l$ & Capacity of link $l$ & \ref{sec_lnctm},\ref{sec_node_model}\\
	$v_l$ & Speed of flow in link $l$ & \ref{sec_lnctm}\\
	$v_l^{f}$ & Freeflow speed of link $l$ & \ref{sec_lnctm} \\
	$w_l$ & Congestion wave speed of link $l$ & \ref{sec_lnctm} \\
	$n_l^J$ & Jam density of link $l$ & \ref{sec_lnctm} \\
	$n_l^{-}$ & Low critical density of link $l$ & \ref{sec_lnctm} \\
	$n_l^{+}$ & High critical density of link $l$ & \ref{sec_lnctm} \\
	$\beta_{ij}^c$ & Split ratio of $c$ vehicles from link $i$ to link $j$ & \ref{sec_lnctm},\ref{sec_node_model},\ref{sec_sr_assignment}\\
	$\eta_{jj'}^i$ & Mutual restriction coefficient of link $j$ onto link $j'$ for link $i$ & \ref{sec_lnctm},\ref{sec_node_model} \\
	$\theta_l$ & Congestion metastate of link $l$ & \ref{sec_lnctm} \\
	$S_l^c$ & Demand for commodity $c$ of link $l$ & \ref{sec_lnctm},\ref{sec_node_model}\\
	$\St_l^c(k)$ & Adjusted demand for commodity $c$ of link $l$ as of iteration $k$ & \ref{sec_node_model} \\
	$R_l$ & Supply of link $l$ & \ref{sec_lnctm},\ref{sec_node_model}\\
	$p_i$ & Priority of input link $i$ & \ref{sec_node_model},\ref{sec_sr_assignment} \\
	$U(k)$ & Set of input links whose flows have yet to be fully determined as of iteration $k$ & \ref{sec_node_model} \\
	$V(k)$ & Set of output links whose flows have yet to be fully determined as of iteration $k$ & \ref{sec_node_model} \\
	$\pt_i(k)$ & Adjusted priority of link $i$ at iteration $k$ & \ref{sec_node_model},\ref{sec_sr_assignment} \\
	$\Rt_j(k)$ & Adjusted supply of link $j$ at iteration $k$ & \ref{sec_node_model} \\
	$U_j(k)$ & Set of input links contributing to link $j$ whose flows are undetermined as of iteration $k$ & \ref{sec_node_model} \\
	$\pt_{ij}(k)$ & Oriented priority from link $i$ to $j$ at iteration $k$ & \ref{sec_node_model},\ref{sec_sr_assignment} \\
	$a_j(k)$ & Restriction term of link $j$ at iteration $k$ & \ref{sec_node_model} \\
	$a_{j^\ast}(k)$ & Smallest (most restrictive) restriction term at iteration $k$ & \ref{sec_node_model} \\
	$\alpha_j(k)$ & Reduction factor of link $j$ at iteration $k$ & \ref{sec_node_model} \\
	$\Ut(k)$ & Set of input links whose demand can be fully met by downstream links at iteration $k$& \ref{sec_node_model} \\
	$\BB$ & Set of commodity movement triples $(i,j,c)$ whose split ratios are known & \ref{sec_sr_assignment} \\
	$\OBB$ & Set of commodity movement triples $(i,j,c)$ whose split ratios are unknown & \ref{sec_sr_assignment} \\
	$\obeta_i^c$ & Unassigned portion of the split ratios of commodity $c$ from link $i$ & \ref{sec_sr_assignment} \\
	$\So_i^c$ & Unassigned demand of commodity $c$ from link $i$ & \ref{sec_sr_assignment} \\
	$\Vt(k)$ & Set of output links to which unassigned demand may still be assigned as of iteration $k$ & \ref{sec_sr_assignment} \\
	$\mu^{+}(k)$ & Largest oriented demand-supply ratio at iteration $k$ & \ref{sec_sr_assignment} \\
	$\mu^{-}(k)$ & Smallest oriented demand-supply ratio at iteration $k$ & \ref{sec_sr_assignment} \\
\end{longtable}

\bibliographystyle{abbrvnat}
\bibliography{traffic}

\providecommand{\url}[1]{#1}
\begin{thebibliography}{14}
\providecommand{\natexlab}[1]{#1}
\providecommand{\url}[1]{\texttt{#1}}
\expandafter\ifx\csname urlstyle\endcsname\relax
  \providecommand{\doi}[1]{doi: #1}\else
  \providecommand{\doi}{doi: \begingroup \urlstyle{rm}\Url}\fi

\bibitem[Bliemer(2007)]{bliemer07}
M.~Bliemer.
\newblock Dynamic queueing and spillback in an analytical multiclass dynamic
  network loading model.
\newblock \emph{Transportation Research Record}, 2029:\penalty0 14--21, 2007.

\bibitem[{\relax California Department of Transportation
  (Caltrans)}(2015)]{csmps}
{\relax California Department of Transportation (Caltrans)}.
\newblock Office of {S}ystem, {F}reight, \& {R}ail {P}lanning, 2015.
\newblock URL \url{http://www.dot.ca.gov/hq/tpp/corridor-mobility/}.

\bibitem[Coifman and Kim(2010)]{capdrop2}
B.~Coifman and S.~Kim.
\newblock Extended bottlenecks, the fundamental relationship, and capacity drop
  on freeways.
\newblock \emph{19th International Symposium of Transportation and Traffic
  Theory}, 2010.

\bibitem[Courant et~al.(1928)Courant, Friedrichs, and Lewy]{cfl28}
R.~Courant, K.~Friedrichs, and H.~Lewy.
\newblock {\"{U}ber die partiellen Differenzengleichungen der mathematischen
  Physik.}
\newblock \emph{Mathematische Annalen}, 100\penalty0 (1):\penalty0 32--74,
  1928.

\bibitem[Daganzo(1994)]{daganzo94}
C.~Daganzo.
\newblock The cell transmission model: A dynamic representation of highway
  traffic consistent with the hydrodynamic theory.
\newblock \emph{Transportation Research, Part B}, 28\penalty0 (4):\penalty0
  269--287, 1994.

\bibitem[Daganzo(1995)]{daganzo95a}
C.~Daganzo.
\newblock The cell transmission model, {Part II}: Network traffic.
\newblock \emph{Transportation Research, Part B}, 29\penalty0 (2):\penalty0
  79--93, 1995.

\bibitem[Gentile et~al.(2007)Gentile, Meschini, and Papola]{gentile07}
G.~Gentile, L.~Meschini, and N.~Papola.
\newblock Spillback congestion in dynamic traffic assignment: a macroscopic
  flow model with time-varying bottlenecks.
\newblock \emph{Transportation Research, Part B}, 41\penalty0 (10):\penalty0
  1114--1138, 2007.

\bibitem[Koshi et~al.(1983)Koshi, Iwasaki, and Ohkura]{koshi1983}
M.~Koshi, M.~Iwasaki, and I.~Ohkura.
\newblock Some findings and an overview on vehicular flow characteristics.
\newblock In \emph{Proceedings of the 8th International Symposium on
  Transportation and Traffic Flow Theory}, volume 198, pages 403--426, 1983.

\bibitem[Lebacque and Khoshyaran(2005)]{lebacque05}
J.~Lebacque and M.~Khoshyaran.
\newblock First-order macroscopic traffic flow models: intersection modeling,
  network modeling.
\newblock In \emph{The 16th International Symposium on Transportation and
  Traffic Theory (ISTTT)}, pages 365--386, 2005.

\bibitem[Saberi and Mahmassani(2013)]{hyst}
M.~Saberi and H.~Mahmassani.
\newblock Hysteresis and capacity drop phenomena in freeway networks: Empirical
  characterization and interpretation.
\newblock \emph{Transportation Research Record: Journal of the Transportation
  Research Board}, 2391:\penalty0 44--55, 2013.

\bibitem[Shiomi et~al.(2015)Shiomi, Taniguchi, Uno, Shimamoto, and
  Nakamura]{shiomi2015}
Y.~Shiomi, T.~Taniguchi, N.~Uno, H.~Shimamoto, and T.~Nakamura.
\newblock Multilane first-order traffic flow model with endogenous
  representation of lane-flow equilibrium.
\newblock \emph{Transportation Research Part C: Emerging Technologies}, In
  press:\penalty0 --, 2015.
\newblock ISSN 0968-090X.
\newblock \doi{http://dx.doi.org/10.1016/j.trc.2015.07.002}.

\bibitem[Srivastava and Gerolimini(2013)]{capdrop}
A.~Srivastava and N.~Gerolimini.
\newblock Empirical observations of capacity drop in freeway merges with ramp
  control and integration in a first-order model.
\newblock \emph{Transportation Research Part C: Emerging Technologies},
  30:\penalty0 161--177, 2013.

\bibitem[Tamp\`{e}re et~al.(2011)Tamp\`{e}re, Corthout, Cattrysse, and
  Immers]{tampere11}
C.~M.~J. Tamp\`{e}re, R.~Corthout, D.~Cattrysse, and L.~H. Immers.
\newblock A generic class of first order node models for dynamic macroscopic
  simulation of traffic flows.
\newblock \emph{Transportation Research, Part B}, 45\penalty0 (1):\penalty0
  289--309, 2011.

\bibitem[Wright and Horowitz(2015)]{pf}
M.~Wright and R.~Horowitz.
\newblock Fusing loop and {GPS} probe measurements to estimate freeway density,
  2015.
\newblock To be submitted to IEEE Transactions on Intelligent Transportation
  Systems.

\end{thebibliography}

\end{document}